\documentclass[letter,11pt,final]{article}
\usepackage[numbers,sort&compress]{natbib}

\usepackage{algorithm}
\usepackage[noend]{algpseudocode}
\usepackage{amsmath,amssymb,amsthm}
\usepackage{bm}
\usepackage{bbm}
\usepackage{bbold}
\usepackage{booktabs}
\usepackage{caption}
\usepackage{enumerate}
\usepackage{float}
\usepackage{graphicx}
\usepackage[colorlinks,citecolor=blue,linkcolor=BrickRed]{hyperref}
\usepackage{cleveref}
\usepackage[utf8]{inputenc}
\usepackage{subcaption}
\usepackage{times}
\usepackage[normalem]{ulem}
\usepackage[usenames,dvipsnames]{xcolor}
\usepackage{xspace}
\usepackage{thmtools}

\theoremstyle{plain}
\declaretheorem[name=Theorem, numberwithin=section]{thm}
\declaretheorem[name=Lemma,sharenumber=thm]{lem}

\declaretheorem[name=Observation,sharenumber=thm]{obs}

\declaretheorem[name=Defintion,sharenumber=thm]{Def}

\newcounter{note}[section]

\newcommand{\bc}[1]{&{\text{$\left(\text{By #1} \right)$}}}

\newcommand{\poly}{\text{poly}}

\newcommand{\mcI}{\mathcal{I}}

\newcommand{\mcS}{\mathcal{S}}
\newcommand{\mcT}{\mathcal{T}}

\allowdisplaybreaks

\usepackage{fullpage}
\usepackage{authblk}
\usepackage{tikz}
\usepackage{pgfplots}
\usepackage{ifthen}
\usepackage{thmtools} 
\usepackage{thm-restate}
\usepackage{MnSymbol}
\usepackage{caption}
\usepackage{anyfontsize}
\usetikzlibrary{positioning}
\usepackage[nomessages]{fp}
\def\edgeWidth{.12}
\tikzstyle{node}=[circle,draw,fill=black, scale=.45, line width=.3mm, minimum size = 12 pt]
\tikzstyle{edgeStyle}=[line width = \edgeWidth cm, font=\tiny]
\usetikzlibrary{calc}

\tikzstyle{nodeSmall}=[circle,draw,fill=black, scale=1, line width=.3mm, minimum size = 12 pt]
\tikzstyle{edgeStyleSmall}=[line width = .2 cm]

\definecolor{c0}{HTML}{FFFFFF}	
\definecolor{S}{HTML}{4363D8}	
\definecolor{T}{HTML}{00E6E6}	
\definecolor{U}{HTML}{FBBC05}	
\definecolor{V}{HTML}{FF8000}	
\definecolor{W}{HTML}{34A853}	
\definecolor{X}{HTML}{B12DD2}	
\definecolor{Y}{HTML}{EA4335}	
\definecolor{Z}{HTML}{800000}	

\definecolor{IDO}{HTML}{FF0000}	
\definecolor{IDT}{HTML}{00FF00}	

\def\drawMultiColorEdge((#1,#2),(#3,#4))[#5][#6][#7][#8]{ 
	\def\clipRadius{2};
	\draw [edgeStyle,color=#5] (#1,#2) -- (#3,#4);
	
	\begin{scope}
		\clip (#1,#2) -- (#3,#4) -- (#3+\clipRadius,#4) -- (#1+\clipRadius, #2) -- cycle;
		\draw [edgeStyle,color=#6] (#1,#2) -- (#3,#4);
	\end{scope}

	\draw[draw=none] (#1,#2) -- (#3,#4) node [midway,fill=white,sloped,pos =#7,font={\fontsize{6pt}{4.2}\selectfont}, inner sep=.5pt,rotate=#8] {\color{#5}#5\color{#6} #6} ;	
}

\def\drawMultiColorEdgeSmall((#1,#2),(#3,#4))[#5][#6][#7][#8]{ 
	\def\clipRadius{2};
	\draw [edgeStyleSmall,color=#5] (#1,#2) -- (#3,#4);
	
	\begin{scope}
		\clip (#1,#2) -- (#3,#4) -- (#3+\clipRadius,#4) -- (#1+\clipRadius, #2) -- cycle;
		\draw [edgeStyleSmall,color=#6] (#1,#2) -- (#3,#4);
	\end{scope}
	
	\draw[draw=none] (#1,#2) -- (#3,#4) node [midway,fill=white,sloped,pos =#7,font={\fontsize{12pt}{4.2}\selectfont}, inner sep=.5pt,rotate=#8] {\color{#5}#5\color{#6} #6} ;	
}

\FPeval{\xInstanceOffset}{1}

\FPeval{\instanceWidth}{3.2}

\FPeval{\totalWidth}{round(3*\xInstanceOffset+4*\instanceWidth, 3)}

\FPeval{\xOneOffset}{\totalWidth/3}
\FPeval{\xTwoOffset}{\totalWidth/2}
\FPeval{\xThreeOffset}{\instanceWidth/3}
\FPeval{\xFourOffset}{\instanceWidth/15}
\FPeval{\secondRootOffset}{\instanceWidth/6}

\def\defineCoordinatesOne{
	
	\def\ya{6}
	\def\yb{5}
	\def\yc{2}
	\def\yd{0}
	
	\FPeval{\xa}{0}
	\FPeval{\xb}{\xa+\xOneOffset}
	\FPeval{\xc}{\xb+\xOneOffset}
	\FPeval{\xd}{\xc+\xOneOffset}
	
	\FPeval{\xe}{\totalWidth/6}
	\FPeval{\xf}{5*\totalWidth/6}
	
	\FPeval{\xgA}{0}	
	\FPeval{\xhA}{\xgA+\xThreeOffset}	
	\FPeval{\xiA}{\xhA+\xThreeOffset}		
	\FPeval{\xjA}{\xiA+\xThreeOffset}	
	
	\FPeval{\xkA}{0}	
	\FPeval{\xlA}{\xkA+\xFourOffset}	
	\FPeval{\xmA}{\xlA+\xFourOffset}
	\FPeval{\xnA}{\xmA+\xFourOffset}
	
	\FPeval{\xoA}{\xnA + \xFourOffset}	
	\FPeval{\xpA}{\xoA + \xFourOffset}	
	\FPeval{\xqA}{\xpA+\xFourOffset}
	\FPeval{\xrA}{\xqA+\xFourOffset}
	
	\FPeval{\xsA}{\xrA+\xFourOffset}	
	\FPeval{\xtA}{\xsA+\xFourOffset}	
	\FPeval{\xuA}{\xtA+\xFourOffset}
	\FPeval{\xvA}{\xuA+\xFourOffset}
	
	\FPeval{\xwA}{\xvA+\xFourOffset}	
	\FPeval{\xxA}{\xwA+\xFourOffset}	
	\FPeval{\xyA}{\xxA+\xFourOffset}
	\FPeval{\xzA}{\xyA+\xFourOffset}
	
	\FPeval{\BOffset}{\instanceWidth+\xInstanceOffset}
	
	\def\xgB{\xgA+\BOffset}	
	\def\xhB{\xhA+\BOffset}	
	\def\xiB{\xiA+\BOffset}		
	\def\xjB{\xjA+\BOffset}	
	
	\def\xkB{\xkA+\BOffset}	
	\def\xlB{\xlA+\BOffset}	
	\def\xmB{\xmA+\BOffset}
	\def\xnB{\xnA+\BOffset}
	
	\def\xoB{\xoA + \BOffset}	
	\def\xpB{\xpA + \BOffset}	
	\def\xqB{\xqA+\BOffset}
	\def\xrB{\xrA+\BOffset}
	
	\def\xsB{\xsA+\BOffset}	
	\def\xtB{\xtA+\BOffset}	
	\def\xuB{\xuA+\BOffset}
	\def\xvB{\xvA+\BOffset}
	
	\def\xwB{\xwA+\BOffset}	
	\def\xxB{\xxA+\BOffset}	
	\def\xyB{\xyA+\BOffset}
	\def\xzB{\xzA+\BOffset}
	
	\def\COffset{2*\instanceWidth+2*\xInstanceOffset}
	
	\def\xgC{\xgA+\COffset}	
	\def\xhC{\xhA+\COffset}	
	\def\xiC{\xiA+\COffset}		
	\def\xjC{\xjA+\COffset}	
	
	\def\xkC{\xkA+\COffset}	
	\def\xlC{\xlA+\COffset}	
	\def\xmC{\xmA+\COffset}
	\def\xnC{\xnA+\COffset}
	
	\def\xoC{\xoA + \COffset}	
	\def\xpC{\xpA + \COffset}	
	\def\xqC{\xqA+\COffset}
	\def\xrC{\xrA+\COffset}
	
	\def\xsC{\xsA+\COffset}	
	\def\xtC{\xtA+\COffset}	
	\def\xuC{\xuA+\COffset}
	\def\xvC{\xvA+\COffset}
	
	\def\xwC{\xwA+\COffset}	
	\def\xxC{\xxA+\COffset}	
	\def\xyC{\xyA+\COffset}
	\def\xzC{\xzA+\COffset}
	
	\def\DOffset{3*\instanceWidth+3*\xInstanceOffset}
	
	\def\xgD{\xgA+\DOffset}	
	\def\xhD{\xhA+\DOffset}	
	\def\xiD{\xiA+\DOffset}		
	\def\xjD{\xjA+\DOffset}	
	
	\def\xkD{\xkA+\DOffset}	
	\def\xlD{\xlA+\DOffset}	
	\def\xmD{\xmA+\DOffset}
	\def\xnD{\xnA+\DOffset}
	
	\def\xoD{\xoA + \DOffset}	
	\def\xpD{\xpA + \DOffset}	
	\def\xqD{\xqA+\DOffset}
	\def\xrD{\xrA+\DOffset}
	
	\def\xsD{\xsA+\DOffset}	
	\def\xtD{\xtA+\DOffset}	
	\def\xuD{\xuA+\DOffset}
	\def\xvD{\xvA+\DOffset}
	
	\def\xwD{\xwA+\DOffset}	
	\def\xxD{\xxA+\DOffset}	
	\def\xyD{\xyA+\DOffset}
	\def\xzD{\xzA+\DOffset}
}

\def \drawNodesOne{
	
	\def\oneOffset{1}
	
	\node at (\xa, \ya+\oneOffset) [node] (aa) {};
	\node at (\xb, \ya+\oneOffset) [node] (ba) {};
	\node at (\xc, \ya+\oneOffset) [node] (ca) {};
	\node at (\xd, \ya+\oneOffset) [node] (da) {};
	
	\node at (\xe, \yb+\oneOffset) [node,fill=white,draw=IDO] (eb) {};
	\node at (\xf, \yb+\oneOffset) [node,fill=white,draw=IDT] (fb) {};
}

\def \drawEdgesOne{
	\def\bottomdist{.87}
	\def\oneOffset{1}
	\drawMultiColorEdge((\xa, \ya+\oneOffset),(\xe, \yb+\oneOffset))[S][T][.5][0]
	\drawMultiColorEdge((\xb, \ya+\oneOffset),(\xe, \yb+\oneOffset))[U][V][.5][0]
	\drawMultiColorEdge((\xc, \ya+\oneOffset),(\xf, \yb+\oneOffset))[W][X][.5][0]
	\drawMultiColorEdge((\xd, \ya+\oneOffset),(\xf, \yb+\oneOffset))[Y][Z][.5][0]
}

\def\drawLBConstructStepOne{
	\begin{scope}
		\defineCoordinatesOne
		\drawEdgesOne
		\drawNodesOne
	\end{scope}
}

\def\defineCoordinatesTwo{
	
	\def\ya{6}
	\def\yb{5}
	\def\yc{2}
	\def\yd{0}
	
	\FPeval{\xa}{0}
	\FPeval{\xb}{\xa+\xOneOffset}
	\FPeval{\xc}{\xb+\xOneOffset}
	\FPeval{\xd}{\xc+\xOneOffset}
	
	\FPeval{\xe}{\totalWidth/6}
	\FPeval{\xf}{5*\totalWidth/6}
	
	\FPeval{\xgA}{0}	
	\FPeval{\xhA}{\xgA+\xThreeOffset}	
	\FPeval{\xiA}{\xhA+\xThreeOffset}		
	\FPeval{\xjA}{\xiA+\xThreeOffset}	
	
	\FPeval{\xkA}{0}	
	\FPeval{\xlA}{\xkA+\xFourOffset}	
	\FPeval{\xmA}{\xlA+\xFourOffset}
	\FPeval{\xnA}{\xmA+\xFourOffset}
	
	\FPeval{\xoA}{\xnA + \xFourOffset}	
	\FPeval{\xpA}{\xoA + \xFourOffset}	
	\FPeval{\xqA}{\xpA+\xFourOffset}
	\FPeval{\xrA}{\xqA+\xFourOffset}
	
	\FPeval{\xsA}{\xrA+\xFourOffset}	
	\FPeval{\xtA}{\xsA+\xFourOffset}	
	\FPeval{\xuA}{\xtA+\xFourOffset}
	\FPeval{\xvA}{\xuA+\xFourOffset}
	
	\FPeval{\xwA}{\xvA+\xFourOffset}	
	\FPeval{\xxA}{\xwA+\xFourOffset}	
	\FPeval{\xyA}{\xxA+\xFourOffset}
	\FPeval{\xzA}{\xyA+\xFourOffset}
	
	\FPeval{\BOffset}{\instanceWidth+\xInstanceOffset}
	
	\def\xgB{\xgA+\BOffset}	
	\def\xhB{\xhA+\BOffset}	
	\def\xiB{\xiA+\BOffset}		
	\def\xjB{\xjA+\BOffset}	
	
	\def\xkB{\xkA+\BOffset}	
	\def\xlB{\xlA+\BOffset}	
	\def\xmB{\xmA+\BOffset}
	\def\xnB{\xnA+\BOffset}
	
	\def\xoB{\xoA + \BOffset}	
	\def\xpB{\xpA + \BOffset}	
	\def\xqB{\xqA+\BOffset}
	\def\xrB{\xrA+\BOffset}
	
	\def\xsB{\xsA+\BOffset}	
	\def\xtB{\xtA+\BOffset}	
	\def\xuB{\xuA+\BOffset}
	\def\xvB{\xvA+\BOffset}
	
	\def\xwB{\xwA+\BOffset}	
	\def\xxB{\xxA+\BOffset}	
	\def\xyB{\xyA+\BOffset}
	\def\xzB{\xzA+\BOffset}
	
	\def\COffset{2*\instanceWidth+2*\xInstanceOffset}
	
	\def\xgC{\xgA+\COffset}	
	\def\xhC{\xhA+\COffset}	
	\def\xiC{\xiA+\COffset}		
	\def\xjC{\xjA+\COffset}	
	
	\def\xkC{\xkA+\COffset}	
	\def\xlC{\xlA+\COffset}	
	\def\xmC{\xmA+\COffset}
	\def\xnC{\xnA+\COffset}
	
	\def\xoC{\xoA + \COffset}	
	\def\xpC{\xpA + \COffset}	
	\def\xqC{\xqA+\COffset}
	\def\xrC{\xrA+\COffset}
	
	\def\xsC{\xsA+\COffset}	
	\def\xtC{\xtA+\COffset}	
	\def\xuC{\xuA+\COffset}
	\def\xvC{\xvA+\COffset}
	
	\def\xwC{\xwA+\COffset}	
	\def\xxC{\xxA+\COffset}	
	\def\xyC{\xyA+\COffset}
	\def\xzC{\xzA+\COffset}
	
	\def\DOffset{3*\instanceWidth+3*\xInstanceOffset}
	
	\def\xgD{\xgA+\DOffset}	
	\def\xhD{\xhA+\DOffset}	
	\def\xiD{\xiA+\DOffset}		
	\def\xjD{\xjA+\DOffset}	
	
	\def\xkD{\xkA+\DOffset}	
	\def\xlD{\xlA+\DOffset}	
	\def\xmD{\xmA+\DOffset}
	\def\xnD{\xnA+\DOffset}
	
	\def\xoD{\xoA + \DOffset}	
	\def\xpD{\xpA + \DOffset}	
	\def\xqD{\xqA+\DOffset}
	\def\xrD{\xrA+\DOffset}
	
	\def\xsD{\xsA+\DOffset}	
	\def\xtD{\xtA+\DOffset}	
	\def\xuD{\xuA+\DOffset}
	\def\xvD{\xvA+\DOffset}
	
	\def\xwD{\xwA+\DOffset}	
	\def\xxD{\xxA+\DOffset}	
	\def\xyD{\xyA+\DOffset}
	\def\xzD{\xzA+\DOffset}
}

\def \drawNodesTwo{
	
	\node at (\xgA, \yb) [node,fill=black,thick,draw=IDO] (eb) {};
	\node at (\xhA, \yb) [node,fill=black,thick,draw=IDO] (eb) {};
	\node at (\xiA, \yb) [node,fill=black,thick,draw=IDO] (eb) {};
	\node at (\xjA, \yb) [node,fill=black,thick,draw=IDO] (eb) {};
	
	\node at (\xgA+\secondRootOffset, \yb) [node,fill=black,thick,draw=IDT] (eb) {};
	\node at (\xhA+\secondRootOffset, \yb) [node,fill=black,thick,draw=IDT] (eb) {};
	\node at (\xiA+\secondRootOffset, \yb) [node,fill=black,thick,draw=IDT] (eb) {};
	\node at (\xjA+\secondRootOffset, \yb) [node,fill=black,thick,draw=IDT] (eb) {};
	
	\node at (\xgA, \yc) [node,fill=white] (gcA) {};
	\node at (\xhA, \yc) [node,fill=white] (hcA) {};
	\node at (\xiA, \yc) [node,fill=white] (icA) {};
	\node at (\xjA, \yc) [node,fill=white] (jcA) {};
	
	\node at (\xkA, \yd) [node,fill=white] (kdA) {};
	\node at (\xlA, \yd) [node,fill=white] (ldA) {};
	\node at (\xmA, \yd) [node,fill=white] (mdA) {};
	\node at (\xnA, \yd) [node,fill=white] (ndA) {};
	
	\node at (\xoA, \yd) [node,fill=white] (odA) {};
	\node at (\xpA, \yd) [node,fill=white] (pdA) {};
	\node at (\xqA, \yd) [node,fill=white] (qdA) {};
	\node at (\xrA, \yd) [node,fill=white] (rdA) {};
	
	\node at (\xsA, \yd) [node,fill=white] (sdA) {};
	\node at (\xtA, \yd) [node,fill=white] (tdA) {};
	\node at (\xuA, \yd) [node,fill=white] (udA) {};
	\node at (\xvA, \yd) [node,fill=white] (vdA) {};
	
	\node at (\xwA, \yd) [node,fill=white] (wdA) {};
	\node at (\xxA, \yd) [node,fill=white] (xdA) {};
	\node at (\xyA, \yd) [node,fill=white] (ydA) {};
	\node at (\xzA, \yd) [node,fill=white] (zdA) {};
	
	\node at (\xgB, \yb) [node,fill=black,thick,draw=IDO] (eb) {};
	\node at (\xhB, \yb) [node,fill=black,thick,draw=IDO] (eb) {};
	\node at (\xiB, \yb) [node,fill=black,thick,draw=IDO] (eb) {};
	\node at (\xjB, \yb) [node,fill=black,thick,draw=IDO] (eb) {};
	
	\node at (\xgB+\secondRootOffset, \yb) [node,fill=black,thick,draw=IDT] (eb) {};
	\node at (\xhB+\secondRootOffset, \yb) [node,fill=black,thick,draw=IDT] (eb) {};
	\node at (\xiB+\secondRootOffset, \yb) [node,fill=black,thick,draw=IDT] (eb) {};
	\node at (\xjB+\secondRootOffset, \yb) [node,fill=black,thick,draw=IDT] (eb) {};
	
	\node at (\xgB, \yc) [node,fill=white] (gcB) {};
	\node at (\xhB, \yc) [node,fill=white] (hcB) {};
	\node at (\xiB, \yc) [node,fill=white] (icB) {};
	\node at (\xjB, \yc) [node,fill=white] (jcB) {};
	
	\node at (\xkB, \yd) [node,fill=white] (kdB) {};
	\node at (\xlB, \yd) [node,fill=white] (ldB) {};
	\node at (\xmB, \yd) [node,fill=white] (mdB) {};
	\node at (\xnB, \yd) [node,fill=white] (ndB) {};
	
	\node at (\xoB, \yd) [node,fill=white] (odB) {};
	\node at (\xpB, \yd) [node,fill=white] (pdB) {};
	\node at (\xqB, \yd) [node,fill=white] (qdB) {};
	\node at (\xrB, \yd) [node,fill=white] (rdB) {};
	
	\node at (\xsB, \yd) [node,fill=white] (sdB) {};
	\node at (\xtB, \yd) [node,fill=white] (tdB) {};
	\node at (\xuB, \yd) [node,fill=white] (udB) {};
	\node at (\xvB, \yd) [node,fill=white] (vdB) {};
	
	\node at (\xwB, \yd) [node,fill=white] (wdB) {};
	\node at (\xxB, \yd) [node,fill=white] (xdB) {};
	\node at (\xyB, \yd) [node,fill=white] (ydB) {};
	\node at (\xzB, \yd) [node,fill=white] (zdB) {};
	
	\node at (\xgC, \yb) [node,fill=black,thick,draw=IDO] (eb) {};
	\node at (\xhC, \yb) [node,fill=black,thick,draw=IDO] (eb) {};
	\node at (\xiC, \yb) [node,fill=black,thick,draw=IDO] (eb) {};
	\node at (\xjC, \yb) [node,fill=black,thick,draw=IDO] (eb) {};
	
	\node at (\xgC+\secondRootOffset, \yb) [node,fill=black,thick,draw=IDT] (eb) {};
	\node at (\xhC+\secondRootOffset, \yb) [node,fill=,thick,draw=IDT] (eb) {};
	\node at (\xiC+\secondRootOffset, \yb) [node,fill=,thick,draw=IDT] (eb) {};
	\node at (\xjC+\secondRootOffset, \yb) [node,fill=,thick,draw=IDT] (eb) {};
	
	\node at (\xgC, \yc) [node,fill=white] (gcC) {};
	\node at (\xhC, \yc) [node,fill=white] (hcC) {};
	\node at (\xiC, \yc) [node,fill=white] (icC) {};
	\node at (\xjC, \yc) [node,fill=white] (jcC) {};
	
	\node at (\xkC, \yd) [node,fill=white] (kdC) {};
	\node at (\xlC, \yd) [node,fill=white] (ldC) {};
	\node at (\xmC, \yd) [node,fill=white] (mdC) {};
	\node at (\xnC, \yd) [node,fill=white] (ndC) {};
	
	\node at (\xoC, \yd) [node,fill=white] (odC) {};
	\node at (\xpC, \yd) [node,fill=white] (pdC) {};
	\node at (\xqC, \yd) [node,fill=white] (qdC) {};
	\node at (\xrC, \yd) [node,fill=white] (rdC) {};
	
	\node at (\xsC, \yd) [node,fill=white] (sdC) {};
	\node at (\xtC, \yd) [node,fill=white] (tdC) {};
	\node at (\xuC, \yd) [node,fill=white] (udC) {};
	\node at (\xvC, \yd) [node,fill=white] (vdC) {};
	
	\node at (\xwC, \yd) [node,fill=white] (wdC) {};
	\node at (\xxC, \yd) [node,fill=white] (xdC) {};
	\node at (\xyC, \yd) [node,fill=white] (ydC) {};
	\node at (\xzC, \yd) [node,fill=white] (zdC) {};
	
	\node at (\xgD, \yb) [node,fill=black,thick,draw=IDO] (eb) {};
	\node at (\xhD, \yb) [node,fill=black,thick,draw=IDO] (eb) {};
	\node at (\xiD, \yb) [node,fill=black,thick,draw=IDO] (eb) {};
	\node at (\xjD, \yb) [node,fill=black,thick,draw=IDO] (eb) {};
	
	\node at (\xgD+\secondRootOffset, \yb) [node,fill=black,thick,draw=IDT] (eb) {};
	\node at (\xhD+\secondRootOffset, \yb) [node,fill=black,thick,draw=IDT] (eb) {};
	\node at (\xiD+\secondRootOffset, \yb) [node,fill=black,thick,draw=IDT] (eb) {};
	\node at (\xjD+\secondRootOffset, \yb) [node,fill=black,thick,draw=IDT] (eb) {};
	
	\node at (\xgD, \yc) [node,fill=white] (gcD) {};
	\node at (\xhD, \yc) [node,fill=white] (hcD) {};
	\node at (\xiD, \yc) [node,fill=white] (icD) {};
	\node at (\xjD, \yc) [node,fill=white] (jcD) {};
	
	\node at (\xkD, \yd) [node,fill=white] (kdD) {};
	\node at (\xlD, \yd) [node,fill=white] (ldD) {};
	\node at (\xmD, \yd) [node,fill=white] (mdD) {};
	\node at (\xnD, \yd) [node,fill=white] (ndD) {};
	
	\node at (\xoD, \yd) [node,fill=white] (odD) {};
	\node at (\xpD, \yd) [node,fill=white] (pdD) {};
	\node at (\xqD, \yd) [node,fill=white] (qdD) {};
	\node at (\xrD, \yd) [node,fill=white] (rdD) {};
	
	\node at (\xsD, \yd) [node,fill=white] (sdD) {};
	\node at (\xtD, \yd) [node,fill=white] (tdD) {};
	\node at (\xuD, \yd) [node,fill=white] (udD) {};
	\node at (\xvD, \yd) [node,fill=white] (vdD) {};
	
	\node at (\xwD, \yd) [node,fill=white] (wdD) {};
	\node at (\xxD, \yd) [node,fill=white] (xdD) {};
	\node at (\xyD, \yd) [node,fill=white] (ydD) {};
	\node at (\xzD, \yd) [node,fill=white] (zdD) {};
}

\def \drawEdgesTwo{
	\def\bottomdist{.87}
	
	\drawMultiColorEdge((\xgA, \yb),(\xgA, \yc))[S][U][.5][180]
	\drawMultiColorEdge((\xhA, \yb),(\xhA, \yc))[S][U][.5][180]
	\drawMultiColorEdge((\xiA, \yb),(\xiA, \yc))[S][U][.5][180]
	\drawMultiColorEdge((\xjA, \yb),(\xjA, \yc))[S][U][.5][180]
	
	\drawMultiColorEdge((\xgA+\secondRootOffset, \yb),(\xgA, \yc))[W][Y][.5][0]
	\drawMultiColorEdge((\xhA+\secondRootOffset, \yb),(\xhA, \yc))[W][Z][.5][0]
	\drawMultiColorEdge((\xiA+\secondRootOffset, \yb),(\xiA, \yc))[X][Y][.5][0]
	\drawMultiColorEdge((\xjA+\secondRootOffset, \yb),(\xjA, \yc))[X][Z][.5][0]
	
	\drawMultiColorEdge((\xgA, \yc),(\xkA, \yd))[S][W][\bottomdist][0]
	\drawMultiColorEdge((\xgA, \yc),(\xlA, \yd))[S][Y][\bottomdist][0]
	\drawMultiColorEdge((\xgA, \yc),(\xmA, \yd))[U][W][\bottomdist][0]
	\drawMultiColorEdge((\xgA, \yc),(\xnA, \yd))[U][Y][\bottomdist][0]
	
	\drawMultiColorEdge((\xhA, \yc),(\xoA, \yd))[S][W][\bottomdist][180]
	\drawMultiColorEdge((\xhA, \yc),(\xpA, \yd))[S][Z][\bottomdist][0]
	\drawMultiColorEdge((\xhA, \yc),(\xqA, \yd))[U][W][\bottomdist][0]
	\drawMultiColorEdge((\xhA, \yc),(\xrA, \yd))[U][Z][\bottomdist][0]
	
	\drawMultiColorEdge((\xiA, \yc),(\xsA, \yd))[S][X][\bottomdist][180]
	\drawMultiColorEdge((\xiA, \yc),(\xtA, \yd))[S][Y][\bottomdist][180]
	\drawMultiColorEdge((\xiA, \yc),(\xuA, \yd))[U][X][\bottomdist][0]
	\drawMultiColorEdge((\xiA, \yc),(\xvA, \yd))[U][Y][\bottomdist][0]
	
	\drawMultiColorEdge((\xjA, \yc),(\xwA, \yd))[S][X][\bottomdist][180]
	\drawMultiColorEdge((\xjA, \yc),(\xxA, \yd))[S][Z][\bottomdist][180]
	\drawMultiColorEdge((\xjA, \yc),(\xyA, \yd))[U][X][\bottomdist][180]
	\drawMultiColorEdge((\xjA, \yc),(\xzA, \yd))[U][Z][\bottomdist][0]
	
	\drawMultiColorEdge((\xgB, \yb),(\xgB, \yc))[S][V][.5][180]
	\drawMultiColorEdge((\xhB, \yb),(\xhB, \yc))[S][V][.5][180]
	\drawMultiColorEdge((\xuB, \yb),(\xiB, \yc))[S][V][.5][180]
	\drawMultiColorEdge((\xjB, \yb),(\xjB, \yc))[S][V][.5][180]
	
	\drawMultiColorEdge((\xgB+\secondRootOffset, \yb),(\xgB, \yc))[W][Y][.5][0]
	\drawMultiColorEdge((\xhB+\secondRootOffset, \yb),(\xhB, \yc))[W][Z][.5][0]
	\drawMultiColorEdge((\xiB+\secondRootOffset, \yb),(\xiB, \yc))[X][Y][.5][0]
	\drawMultiColorEdge((\xjB+\secondRootOffset, \yb),(\xjB, \yc))[X][Z][.5][0]
	
	\drawMultiColorEdge((\xgB, \yc),(\xkB, \yd))[S][W][\bottomdist][0]
	\drawMultiColorEdge((\xgB, \yc),(\xlB, \yd))[S][Y][\bottomdist][0]
	\drawMultiColorEdge((\xgB, \yc),(\xmB, \yd))[V][W][\bottomdist][0]
	\drawMultiColorEdge((\xgB, \yc),(\xnB, \yd))[V][Y][\bottomdist][0]
	
	\drawMultiColorEdge((\xhB, \yc),(\xoB, \yd))[S][W][\bottomdist][180]
	\drawMultiColorEdge((\xhB, \yc),(\xpB, \yd))[S][Z][\bottomdist][0]
	\drawMultiColorEdge((\xhB, \yc),(\xqB, \yd))[V][W][\bottomdist][0]
	\drawMultiColorEdge((\xhB, \yc),(\xrB, \yd))[V][Z][\bottomdist][0]
	
	\drawMultiColorEdge((\xiB, \yc),(\xsB, \yd))[S][X][\bottomdist][180]
	\drawMultiColorEdge((\xiB, \yc),(\xtB, \yd))[S][Y][\bottomdist][180]
	\drawMultiColorEdge((\xiB, \yc),(\xuB, \yd))[V][X][\bottomdist][0]
	\drawMultiColorEdge((\xiB, \yc),(\xvB, \yd))[V][Y][\bottomdist][0]
	
	\drawMultiColorEdge((\xjB, \yc),(\xwB, \yd))[S][X][\bottomdist][180]
	\drawMultiColorEdge((\xjB, \yc),(\xxB, \yd))[S][Z][\bottomdist][180]
	\drawMultiColorEdge((\xjB, \yc),(\xyB, \yd))[V][X][\bottomdist][180]
	\drawMultiColorEdge((\xjB, \yc),(\xzB, \yd))[V][Z][\bottomdist][0]
	
	\drawMultiColorEdge((\xgC, \yb),(\xgC, \yc))[T][U][.5][180]
	\drawMultiColorEdge((\xhC, \yb),(\xhC, \yc))[T][U][.5][180]
	\drawMultiColorEdge((\xiC, \yb),(\xiC, \yc))[T][U][.5][180]
	\drawMultiColorEdge((\xjC, \yb),(\xjC, \yc))[T][U][.5][180]
	
	\drawMultiColorEdge((\xgC+\secondRootOffset, \yb),(\xgC, \yc))[W][Y][.5][0]
	\drawMultiColorEdge((\xhC+\secondRootOffset, \yb),(\xhC, \yc))[W][Z][.5][0]
	\drawMultiColorEdge((\xiC+\secondRootOffset, \yb),(\xiC, \yc))[X][Y][.5][0]
	\drawMultiColorEdge((\xjC+\secondRootOffset, \yb),(\xjC, \yc))[X][Z][.5][0]
	
	\drawMultiColorEdge((\xgC, \yc),(\xkC, \yd))[T][W][\bottomdist][0]
	\drawMultiColorEdge((\xgC, \yc),(\xlC, \yd))[T][Y][\bottomdist][0]
	\drawMultiColorEdge((\xgC, \yc),(\xmC, \yd))[U][W][\bottomdist][0]
	\drawMultiColorEdge((\xgC, \yc),(\xnC, \yd))[U][Y][\bottomdist][0]
	
	\drawMultiColorEdge((\xhC, \yc),(\xoC, \yd))[T][W][\bottomdist][180]
	\drawMultiColorEdge((\xhC, \yc),(\xpC, \yd))[T][Z][\bottomdist][0]
	\drawMultiColorEdge((\xhC, \yc),(\xqC, \yd))[U][W][\bottomdist][0]
	\drawMultiColorEdge((\xhC, \yc),(\xrC, \yd))[U][Z][\bottomdist][0]
	
	\drawMultiColorEdge((\xiC, \yc),(\xsC, \yd))[T][X][\bottomdist][180]
	\drawMultiColorEdge((\xiC, \yc),(\xtC, \yd))[T][Y][\bottomdist][180]
	\drawMultiColorEdge((\xiC, \yc),(\xuC, \yd))[U][X][\bottomdist][0]
	\drawMultiColorEdge((\xiC, \yc),(\xvC, \yd))[U][Y][\bottomdist][0]
	
	\drawMultiColorEdge((\xjC, \yc),(\xwC, \yd))[T][X][\bottomdist][180]
	\drawMultiColorEdge((\xjC, \yc),(\xxC, \yd))[T][Z][\bottomdist][180]
	\drawMultiColorEdge((\xjC, \yc),(\xyC, \yd))[U][X][\bottomdist][180]
	\drawMultiColorEdge((\xjC, \yc),(\xzC, \yd))[U][Z][\bottomdist][0]
	
	\drawMultiColorEdge((\xgD, \yb),(\xgD, \yc))[T][V][.5][180]
	\drawMultiColorEdge((\xhD, \yb),(\xhD, \yc))[T][V][.5][180]
	\drawMultiColorEdge((\xiD, \yb),(\xiD, \yc))[T][V][.5][180]
	\drawMultiColorEdge((\xjD, \yb),(\xjD, \yc))[T][V][.5][180]
	
	\drawMultiColorEdge((\xgD+\secondRootOffset, \yb),(\xgD, \yc))[W][Y][.5][0]
	\drawMultiColorEdge((\xhD+\secondRootOffset, \yb),(\xhD, \yc))[W][Z][.5][0]
	\drawMultiColorEdge((\xiD+\secondRootOffset, \yb),(\xiD, \yc))[X][Y][.5][0]
	\drawMultiColorEdge((\xjD+\secondRootOffset, \yb),(\xjD, \yc))[X][Z][.5][0]
	
	\drawMultiColorEdge((\xgD, \yc),(\xkD, \yd))[T][W][\bottomdist][0]
	\drawMultiColorEdge((\xgD, \yc),(\xlD, \yd))[T][Y][\bottomdist][0]
	\drawMultiColorEdge((\xgD, \yc),(\xmD, \yd))[V][W][\bottomdist][0]
	\drawMultiColorEdge((\xgD, \yc),(\xnD, \yd))[V][Y][\bottomdist][0]
	
	\drawMultiColorEdge((\xhD, \yc),(\xoD, \yd))[T][W][\bottomdist][180]
	\drawMultiColorEdge((\xhD, \yc),(\xpD, \yd))[T][Z][\bottomdist][0]
	\drawMultiColorEdge((\xhD, \yc),(\xqD, \yd))[V][W][\bottomdist][0]
	\drawMultiColorEdge((\xhD, \yc),(\xrD, \yd))[V][Z][\bottomdist][0]
	
	\drawMultiColorEdge((\xiD, \yc),(\xsD, \yd))[T][X][\bottomdist][180]
	\drawMultiColorEdge((\xiD, \yc),(\xtD, \yd))[T][Y][\bottomdist][180]
	\drawMultiColorEdge((\xiD, \yc),(\xuD, \yd))[V][X][\bottomdist][0]
	\drawMultiColorEdge((\xiD, \yc),(\xvD, \yd))[V][Y][\bottomdist][0]
	
	\drawMultiColorEdge((\xjD, \yc),(\xwD, \yd))[T][X][\bottomdist][180]
	\drawMultiColorEdge((\xjD, \yc),(\xxD, \yd))[T][Z][\bottomdist][180]
	\drawMultiColorEdge((\xjD, \yc),(\xyD, \yd))[V][X][\bottomdist][180]
	\drawMultiColorEdge((\xjD, \yc),(\xzD, \yd))[V][Z][\bottomdist][0]
}

\def\drawLBConstructStepTwo{
	\begin{scope}
		\defineCoordinatesOne
		\drawEdgesOne
		\drawNodesOne
	\end{scope}
	
	\begin{scope}
		\defineCoordinatesTwo
		\drawEdgesTwo
		\drawNodesTwo
	\end{scope}
}

\def\defineCoordinatesThree{

	\def\ya{6}
	\def\yb{5}
	\def\yc{2}
	\def\yd{0}

	\FPeval{\xa}{0}
	\FPeval{\xb}{\xa+\xOneOffset}
	\FPeval{\xc}{\xb+\xOneOffset}
	\FPeval{\xd}{\xc+\xOneOffset}
	
	\FPeval{\xe}{\totalWidth/6}
	\FPeval{\xf}{5*\totalWidth/6}
	
	\FPeval{\xgA}{0}	
	\FPeval{\xhA}{\xgA+\xThreeOffset}	
	\FPeval{\xiA}{\xhA+\xThreeOffset}		
	\FPeval{\xjA}{\xiA+\xThreeOffset}	
	
	\FPeval{\xkA}{0}	
	\FPeval{\xlA}{\xkA+\xFourOffset}	
	\FPeval{\xmA}{\xlA+\xFourOffset}
	\FPeval{\xnA}{\xmA+\xFourOffset}
	
	\FPeval{\xoA}{\xnA + \xFourOffset}	
	\FPeval{\xpA}{\xoA + \xFourOffset}	
	\FPeval{\xqA}{\xpA+\xFourOffset}
	\FPeval{\xrA}{\xqA+\xFourOffset}
	
	\FPeval{\xsA}{\xrA+\xFourOffset}	
	\FPeval{\xtA}{\xsA+\xFourOffset}	
	\FPeval{\xuA}{\xtA+\xFourOffset}
	\FPeval{\xvA}{\xuA+\xFourOffset}
	
	\FPeval{\xwA}{\xvA+\xFourOffset}	
	\FPeval{\xxA}{\xwA+\xFourOffset}	
	\FPeval{\xyA}{\xxA+\xFourOffset}
	\FPeval{\xzA}{\xyA+\xFourOffset}
	
	\FPeval{\BOffset}{\instanceWidth+\xInstanceOffset}
	
	\def\xgB{\xgA+\BOffset}	
	\def\xhB{\xhA+\BOffset}	
	\def\xiB{\xiA+\BOffset}		
	\def\xjB{\xjA+\BOffset}	
	
	\def\xkB{\xkA+\BOffset}	
	\def\xlB{\xlA+\BOffset}	
	\def\xmB{\xmA+\BOffset}
	\def\xnB{\xnA+\BOffset}
	
	\def\xoB{\xoA + \BOffset}	
	\def\xpB{\xpA + \BOffset}	
	\def\xqB{\xqA+\BOffset}
	\def\xrB{\xrA+\BOffset}
	
	\def\xsB{\xsA+\BOffset}	
	\def\xtB{\xtA+\BOffset}	
	\def\xuB{\xuA+\BOffset}
	\def\xvB{\xvA+\BOffset}
	
	\def\xwB{\xwA+\BOffset}	
	\def\xxB{\xxA+\BOffset}	
	\def\xyB{\xyA+\BOffset}
	\def\xzB{\xzA+\BOffset}
	
	\def\COffset{2*\instanceWidth+2*\xInstanceOffset}
	
	\def\xgC{\xgA+\COffset}	
	\def\xhC{\xhA+\COffset}	
	\def\xiC{\xiA+\COffset}		
	\def\xjC{\xjA+\COffset}	
	
	\def\xkC{\xkA+\COffset}	
	\def\xlC{\xlA+\COffset}	
	\def\xmC{\xmA+\COffset}
	\def\xnC{\xnA+\COffset}
	
	\def\xoC{\xoA + \COffset}	
	\def\xpC{\xpA + \COffset}	
	\def\xqC{\xqA+\COffset}
	\def\xrC{\xrA+\COffset}
	
	\def\xsC{\xsA+\COffset}	
	\def\xtC{\xtA+\COffset}	
	\def\xuC{\xuA+\COffset}
	\def\xvC{\xvA+\COffset}
	
	\def\xwC{\xwA+\COffset}	
	\def\xxC{\xxA+\COffset}	
	\def\xyC{\xyA+\COffset}
	\def\xzC{\xzA+\COffset}
	
	\def\DOffset{3*\instanceWidth+3*\xInstanceOffset}
	
	\def\xgD{\xgA+\DOffset}	
	\def\xhD{\xhA+\DOffset}	
	\def\xiD{\xiA+\DOffset}		
	\def\xjD{\xjA+\DOffset}	
	
	\def\xkD{\xkA+\DOffset}	
	\def\xlD{\xlA+\DOffset}	
	\def\xmD{\xmA+\DOffset}
	\def\xnD{\xnA+\DOffset}
	
	\def\xoD{\xoA + \DOffset}	
	\def\xpD{\xpA + \DOffset}	
	\def\xqD{\xqA+\DOffset}
	\def\xrD{\xrA+\DOffset}
	
	\def\xsD{\xsA+\DOffset}	
	\def\xtD{\xtA+\DOffset}	
	\def\xuD{\xuA+\DOffset}
	\def\xvD{\xvA+\DOffset}
	
	\def\xwD{\xwA+\DOffset}	
	\def\xxD{\xxA+\DOffset}	
	\def\xyD{\xyA+\DOffset}
	\def\xzD{\xzA+\DOffset}
}

\def \drawNodesThree{
	
	\node at (\xa, \ya) [node] (aa) {};
	\node at (\xb, \ya) [node] (ba) {};
	\node at (\xc, \ya) [node] (ca) {};
	\node at (\xd, \ya) [node] (da) {};
	
	\node at (\xe, \yb) [node,fill=white] (eb) {};
	\node at (\xf, \yb) [node,fill=white] (fb) {};
	
	\node at (\xgA, \yc) [node,fill=white] (gcA) {};
	\node at (\xhA, \yc) [node,fill=white] (hcA) {};
	\node at (\xiA, \yc) [node,fill=white] (icA) {};
	\node at (\xjA, \yc) [node,fill=white] (jcA) {};
	
	\node at (\xkA, \yd) [node,fill=white] (kdA) {};
	\node at (\xlA, \yd) [node,fill=white] (ldA) {};
	\node at (\xmA, \yd) [node,fill=white] (mdA) {};
	\node at (\xnA, \yd) [node,fill=white] (ndA) {};
	
	\node at (\xoA, \yd) [node,fill=white] (odA) {};
	\node at (\xpA, \yd) [node,fill=white] (pdA) {};
	\node at (\xqA, \yd) [node,fill=white] (qdA) {};
	\node at (\xrA, \yd) [node,fill=white] (rdA) {};
	
	\node at (\xsA, \yd) [node,fill=white] (sdA) {};
	\node at (\xtA, \yd) [node,fill=white] (tdA) {};
	\node at (\xuA, \yd) [node,fill=white] (udA) {};
	\node at (\xvA, \yd) [node,fill=white] (vdA) {};
	
	\node at (\xwA, \yd) [node,fill=white] (wdA) {};
	\node at (\xxA, \yd) [node,fill=white] (xdA) {};
	\node at (\xyA, \yd) [node,fill=white] (ydA) {};
	\node at (\xzA, \yd) [node,fill=white] (zdA) {};
	
	\node at (\xgB, \yc) [node,fill=white] (gcB) {};
	\node at (\xhB, \yc) [node,fill=white] (hcB) {};
	\node at (\xiB, \yc) [node,fill=white] (icB) {};
	\node at (\xjB, \yc) [node,fill=white] (jcB) {};
	
	\node at (\xkB, \yd) [node,fill=white] (kdB) {};
	\node at (\xlB, \yd) [node,fill=white] (ldB) {};
	\node at (\xmB, \yd) [node,fill=white] (mdB) {};
	\node at (\xnB, \yd) [node,fill=white] (ndB) {};
	
	\node at (\xoB, \yd) [node,fill=white] (odB) {};
	\node at (\xpB, \yd) [node,fill=white] (pdB) {};
	\node at (\xqB, \yd) [node,fill=white] (qdB) {};
	\node at (\xrB, \yd) [node,fill=white] (rdB) {};
	
	\node at (\xsB, \yd) [node,fill=white] (sdB) {};
	\node at (\xtB, \yd) [node,fill=white] (tdB) {};
	\node at (\xuB, \yd) [node,fill=white] (udB) {};
	\node at (\xvB, \yd) [node,fill=white] (vdB) {};
	
	\node at (\xwB, \yd) [node,fill=white] (wdB) {};
	\node at (\xxB, \yd) [node,fill=white] (xdB) {};
	\node at (\xyB, \yd) [node,fill=white] (ydB) {};
	\node at (\xzB, \yd) [node,fill=white] (zdB) {};
	
	\node at (\xgC, \yc) [node,fill=white] (gcC) {};
	\node at (\xhC, \yc) [node,fill=white] (hcC) {};
	\node at (\xiC, \yc) [node,fill=white] (icC) {};
	\node at (\xjC, \yc) [node,fill=white] (jcC) {};
	
	\node at (\xkC, \yd) [node,fill=white] (kdC) {};
	\node at (\xlC, \yd) [node,fill=white] (ldC) {};
	\node at (\xmC, \yd) [node,fill=white] (mdC) {};
	\node at (\xnC, \yd) [node,fill=white] (ndC) {};
	
	\node at (\xoC, \yd) [node,fill=white] (odC) {};
	\node at (\xpC, \yd) [node,fill=white] (pdC) {};
	\node at (\xqC, \yd) [node,fill=white] (qdC) {};
	\node at (\xrC, \yd) [node,fill=white] (rdC) {};
	
	\node at (\xsC, \yd) [node,fill=white] (sdC) {};
	\node at (\xtC, \yd) [node,fill=white] (tdC) {};
	\node at (\xuC, \yd) [node,fill=white] (udC) {};
	\node at (\xvC, \yd) [node,fill=white] (vdC) {};
	
	\node at (\xwC, \yd) [node,fill=white] (wdC) {};
	\node at (\xxC, \yd) [node,fill=white] (xdC) {};
	\node at (\xyC, \yd) [node,fill=white] (ydC) {};
	\node at (\xzC, \yd) [node,fill=white] (zdC) {};
	
	\node at (\xgD, \yc) [node,fill=white] (gcD) {};
	\node at (\xhD, \yc) [node,fill=white] (hcD) {};
	\node at (\xiD, \yc) [node,fill=white] (icD) {};
	\node at (\xjD, \yc) [node,fill=white] (jcD) {};
	
	\node at (\xkD, \yd) [node,fill=white] (kdD) {};
	\node at (\xlD, \yd) [node,fill=white] (ldD) {};
	\node at (\xmD, \yd) [node,fill=white] (mdD) {};
	\node at (\xnD, \yd) [node,fill=white] (ndD) {};
	
	\node at (\xoD, \yd) [node,fill=white] (odD) {};
	\node at (\xpD, \yd) [node,fill=white] (pdD) {};
	\node at (\xqD, \yd) [node,fill=white] (qdD) {};
	\node at (\xrD, \yd) [node,fill=white] (rdD) {};
	
	\node at (\xsD, \yd) [node,fill=white] (sdD) {};
	\node at (\xtD, \yd) [node,fill=white] (tdD) {};
	\node at (\xuD, \yd) [node,fill=white] (udD) {};
	\node at (\xvD, \yd) [node,fill=white] (vdD) {};
	
	\node at (\xwD, \yd) [node,fill=white] (wdD) {};
	\node at (\xxD, \yd) [node,fill=white] (xdD) {};
	\node at (\xyD, \yd) [node,fill=white] (ydD) {};
	\node at (\xzD, \yd) [node,fill=white] (zdD) {};
}

\def \drawEdgesThree{
	\def\bottomdist{.87}
	
	\drawMultiColorEdge((\xa, \ya),(\xe, \yb))[S][T][.5][0]
	\drawMultiColorEdge((\xb, \ya),(\xe, \yb))[U][V][.5][0]
	\drawMultiColorEdge((\xc, \ya),(\xf, \yb))[W][X][.5][0]
	\drawMultiColorEdge((\xd, \ya),(\xf, \yb))[Y][Z][.5][0]
	
	\drawMultiColorEdge((\xe, \yb),(\xgA, \yc))[S][U][.5][180]
	\drawMultiColorEdge((\xe, \yb),(\xhA, \yc))[S][U][.5][180]
	\drawMultiColorEdge((\xe, \yb),(\xiA, \yc))[S][U][.5][180]
	\drawMultiColorEdge((\xe, \yb),(\xjA, \yc))[S][U][.5][0]
	
	\drawMultiColorEdge((\xf, \yb),(\xgA, \yc))[W][Y][.95][0]
	\drawMultiColorEdge((\xf, \yb),(\xhA, \yc))[W][Z][.95][0]
	\drawMultiColorEdge((\xf, \yb),(\xiA, \yc))[X][Y][.95][0]
	\drawMultiColorEdge((\xf, \yb),(\xjA, \yc))[X][Z][.965][0]
	
	\drawMultiColorEdge((\xgA, \yc),(\xkA, \yd))[S][W][\bottomdist][0]
	\drawMultiColorEdge((\xgA, \yc),(\xlA, \yd))[S][Y][\bottomdist][0]
	\drawMultiColorEdge((\xgA, \yc),(\xmA, \yd))[U][W][\bottomdist][0]
	\drawMultiColorEdge((\xgA, \yc),(\xnA, \yd))[U][Y][\bottomdist][0]
	
	\drawMultiColorEdge((\xhA, \yc),(\xoA, \yd))[S][W][\bottomdist][180]
	\drawMultiColorEdge((\xhA, \yc),(\xpA, \yd))[S][Z][\bottomdist][0]
	\drawMultiColorEdge((\xhA, \yc),(\xqA, \yd))[U][W][\bottomdist][0]
	\drawMultiColorEdge((\xhA, \yc),(\xrA, \yd))[U][Z][\bottomdist][0]
	
	\drawMultiColorEdge((\xiA, \yc),(\xsA, \yd))[S][X][\bottomdist][180]
	\drawMultiColorEdge((\xiA, \yc),(\xtA, \yd))[S][Y][\bottomdist][180]
	\drawMultiColorEdge((\xiA, \yc),(\xuA, \yd))[U][X][\bottomdist][0]
	\drawMultiColorEdge((\xiA, \yc),(\xvA, \yd))[U][Y][\bottomdist][0]
	
	\drawMultiColorEdge((\xjA, \yc),(\xwA, \yd))[S][X][\bottomdist][180]
	\drawMultiColorEdge((\xjA, \yc),(\xxA, \yd))[S][Z][\bottomdist][180]
	\drawMultiColorEdge((\xjA, \yc),(\xyA, \yd))[U][X][\bottomdist][180]
	\drawMultiColorEdge((\xjA, \yc),(\xzA, \yd))[U][Z][\bottomdist][0]
	
	\drawMultiColorEdge((\xe, \yb),(\xgB, \yc))[S][V][.5][0]
	\drawMultiColorEdge((\xe, \yb),(\xhB, \yc))[S][V][.5][0]
	\drawMultiColorEdge((\xe, \yb),(\xiB, \yc))[S][V][.5][0]
	\drawMultiColorEdge((\xe, \yb),(\xjB, \yc))[S][V][.45][0]
	
	\drawMultiColorEdge((\xf, \yb),(\xgB, \yc))[W][Y][.95][0]
	\drawMultiColorEdge((\xf, \yb),(\xhB, \yc))[W][Z][.95][0]
	\drawMultiColorEdge((\xf, \yb),(\xiB, \yc))[X][Y][.95][0]
	\drawMultiColorEdge((\xf, \yb),(\xjB, \yc))[X][Z][.45][0]
	
	\drawMultiColorEdge((\xgB, \yc),(\xkB, \yd))[S][W][\bottomdist][0]
	\drawMultiColorEdge((\xgB, \yc),(\xlB, \yd))[S][Y][\bottomdist][0]
	\drawMultiColorEdge((\xgB, \yc),(\xmB, \yd))[V][W][\bottomdist][0]
	\drawMultiColorEdge((\xgB, \yc),(\xnB, \yd))[V][Y][\bottomdist][0]
	
	\drawMultiColorEdge((\xhB, \yc),(\xoB, \yd))[S][W][\bottomdist][180]
	\drawMultiColorEdge((\xhB, \yc),(\xpB, \yd))[S][Z][\bottomdist][0]
	\drawMultiColorEdge((\xhB, \yc),(\xqB, \yd))[V][W][\bottomdist][0]
	\drawMultiColorEdge((\xhB, \yc),(\xrB, \yd))[V][Z][\bottomdist][0]
	
	\drawMultiColorEdge((\xiB, \yc),(\xsB, \yd))[S][X][\bottomdist][180]
	\drawMultiColorEdge((\xiB, \yc),(\xtB, \yd))[S][Y][\bottomdist][180]
	\drawMultiColorEdge((\xiB, \yc),(\xuB, \yd))[V][X][\bottomdist][0]
	\drawMultiColorEdge((\xiB, \yc),(\xvB, \yd))[V][Y][\bottomdist][0]
	
	\drawMultiColorEdge((\xjB, \yc),(\xwB, \yd))[S][X][\bottomdist][180]
	\drawMultiColorEdge((\xjB, \yc),(\xxB, \yd))[S][Z][\bottomdist][180]
	\drawMultiColorEdge((\xjB, \yc),(\xyB, \yd))[V][X][\bottomdist][180]
	\drawMultiColorEdge((\xjB, \yc),(\xzB, \yd))[V][Z][\bottomdist][0]
	
	\drawMultiColorEdge((\xe, \yb),(\xgC, \yc))[T][U][.9][0]
	\drawMultiColorEdge((\xe, \yb),(\xhC, \yc))[T][U][.95][0]
	\drawMultiColorEdge((\xe, \yb),(\xiC, \yc))[T][U][.95][0]
	\drawMultiColorEdge((\xe, \yb),(\xjC, \yc))[T][U][.95][0]
	
	\drawMultiColorEdge((\xf, \yb),(\xgC, \yc))[W][Y][.5][0]
	\drawMultiColorEdge((\xf, \yb),(\xhC, \yc))[W][Z][.5][0]
	\drawMultiColorEdge((\xf, \yb),(\xiC, \yc))[X][Y][.5][0]
	\drawMultiColorEdge((\xf, \yb),(\xjC, \yc))[X][Z][.5][0]
	
	\drawMultiColorEdge((\xgC, \yc),(\xkC, \yd))[T][W][\bottomdist][0]
	\drawMultiColorEdge((\xgC, \yc),(\xlC, \yd))[T][Y][\bottomdist][0]
	\drawMultiColorEdge((\xgC, \yc),(\xmC, \yd))[U][W][\bottomdist][0]
	\drawMultiColorEdge((\xgC, \yc),(\xnC, \yd))[U][Y][\bottomdist][0]
	
	\drawMultiColorEdge((\xhC, \yc),(\xoC, \yd))[T][W][\bottomdist][180]
	\drawMultiColorEdge((\xhC, \yc),(\xpC, \yd))[T][Z][\bottomdist][0]
	\drawMultiColorEdge((\xhC, \yc),(\xqC, \yd))[U][W][\bottomdist][0]
	\drawMultiColorEdge((\xhC, \yc),(\xrC, \yd))[U][Z][\bottomdist][0]
	
	\drawMultiColorEdge((\xiC, \yc),(\xsC, \yd))[T][X][\bottomdist][180]
	\drawMultiColorEdge((\xiC, \yc),(\xtC, \yd))[T][Y][\bottomdist][180]
	\drawMultiColorEdge((\xiC, \yc),(\xuC, \yd))[U][X][\bottomdist][0]
	\drawMultiColorEdge((\xiC, \yc),(\xvC, \yd))[U][Y][\bottomdist][0]
	
	\drawMultiColorEdge((\xjC, \yc),(\xwC, \yd))[T][X][\bottomdist][180]
	\drawMultiColorEdge((\xjC, \yc),(\xxC, \yd))[T][Z][\bottomdist][180]
	\drawMultiColorEdge((\xjC, \yc),(\xyC, \yd))[U][X][\bottomdist][180]
	\drawMultiColorEdge((\xjC, \yc),(\xzC, \yd))[U][Z][\bottomdist][0]
	
	\drawMultiColorEdge((\xe, \yb),(\xgD, \yc))[T][V][.965][0]
	\drawMultiColorEdge((\xe, \yb),(\xhD, \yc))[T][V][.95][0]
	\drawMultiColorEdge((\xe, \yb),(\xiD, \yc))[T][V][.95][0]
	\drawMultiColorEdge((\xe, \yb),(\xjD, \yc))[T][V][.95][0]
	
	\drawMultiColorEdge((\xf, \yb),(\xgD, \yc))[W][Y][.5][180]
	\drawMultiColorEdge((\xf, \yb),(\xhD, \yc))[W][Z][.5][0]
	\drawMultiColorEdge((\xf, \yb),(\xiD, \yc))[X][Y][.5][0]
	\drawMultiColorEdge((\xf, \yb),(\xjD, \yc))[X][Z][.5][0]
	
	\drawMultiColorEdge((\xgD, \yc),(\xkD, \yd))[T][W][\bottomdist][0]
	\drawMultiColorEdge((\xgD, \yc),(\xlD, \yd))[T][Y][\bottomdist][0]
	\drawMultiColorEdge((\xgD, \yc),(\xmD, \yd))[V][W][\bottomdist][0]
	\drawMultiColorEdge((\xgD, \yc),(\xnD, \yd))[V][Y][\bottomdist][0]
	
	\drawMultiColorEdge((\xhD, \yc),(\xoD, \yd))[T][W][\bottomdist][180]
	\drawMultiColorEdge((\xhD, \yc),(\xpD, \yd))[T][Z][\bottomdist][0]
	\drawMultiColorEdge((\xhD, \yc),(\xqD, \yd))[V][W][\bottomdist][0]
	\drawMultiColorEdge((\xhD, \yc),(\xrD, \yd))[V][Z][\bottomdist][0]
	
	\drawMultiColorEdge((\xiD, \yc),(\xsD, \yd))[T][X][\bottomdist][180]
	\drawMultiColorEdge((\xiD, \yc),(\xtD, \yd))[T][Y][\bottomdist][180]
	\drawMultiColorEdge((\xiD, \yc),(\xuD, \yd))[V][X][\bottomdist][0]
	\drawMultiColorEdge((\xiD, \yc),(\xvD, \yd))[V][Y][\bottomdist][0]
	
	\drawMultiColorEdge((\xjD, \yc),(\xwD, \yd))[T][X][\bottomdist][180]
	\drawMultiColorEdge((\xjD, \yc),(\xxD, \yd))[T][Z][\bottomdist][180]
	\drawMultiColorEdge((\xjD, \yc),(\xyD, \yd))[V][X][\bottomdist][180]
	\drawMultiColorEdge((\xjD, \yc),(\xzD, \yd))[V][Z][\bottomdist][0]
}

\def\drawLBConstructStepThree{
	\begin{scope}
	\defineCoordinatesThree
	\drawEdgesThree
	\drawNodesThree
	\end{scope}
}

\def \drawNodesSmallNoGuess{
	\node at (\xa, \yaS) [nodeSmall,text=white,inner sep=0pt] (aa) {\small $r_1$};
	\node at (\xb, \yaS) [nodeSmall,text=white, inner sep=0pt] (ba) {\small $r_2$};
	
	\node at (\xe, \ybS) [nodeSmall,fill=white,inner sep=0pt] (eb) {\small $v$};
	
	\node at (\xg, \ycS) [nodeSmall,fill=white,inner sep=0pt] (gc) {\small $u$};
}

\def \drawNodesSmallGuess{
	\node at (\xa, \yaS) [nodeSmall,text=white,inner sep=0pt] (aa) {\small $r_1$};
	\node at (\xb, \yaS) [nodeSmall,text=white, inner sep=0pt] (ba) {\small $r_2$};
	
	\node at (\xe, \ybS) [nodeSmall,fill=white,inner sep=0pt] (eb) {\small $v$};
	
	\node at (\xg, \ycS) [nodeSmall,fill=white,inner sep=0pt] (gc) {\small $u'$};
	\node at (\xh, \ycS) [nodeSmall,fill=white] (hc) {};
	\node at (\xi, \ycS) [nodeSmall,fill=white] (ic) {};
	\node at (\xj, \ycS) [nodeSmall,fill=white] (jc) {};
}

\def \drawEdgesSmallNoGuess{
	\drawMultiColorEdgeSmall((\xa, \yaS),(\xe, \ybS))[S][T][.5][0]
	\drawMultiColorEdgeSmall((\xb, \yaS),(\xe, \ybS))[U][V][.5][0]
	
	\drawMultiColorEdgeSmall((\xg, \ycS),(\xe, \ybS))[S][U][.5][0]
}

\def \drawEdgesSmallGuess{
	\drawMultiColorEdgeSmall((\xa, \yaS),(\xe, \ybS))[S][T][.5][0]
	\drawMultiColorEdgeSmall((\xb, \yaS),(\xe, \ybS))[U][V][.5][0]
	
	\drawMultiColorEdgeSmall((\xe, \ybS),(\xg, \ycS))[S][U][.5][0]
	\drawMultiColorEdgeSmall((\xe, \ybS),(\xh, \ycS))[S][V][.5][0]
	\drawMultiColorEdgeSmall((\xe, \ybS),(\xi, \ycS))[T][U][.5][0]
	\drawMultiColorEdgeSmall((\xe, \ybS),(\xj, \ycS))[T][V][.5][0]
}

\def \drawLBFigSmallNoGuess{
	\begin{scope}
	\def\yaS{10}
	\def\ybS{9}
	\def\ycS{7}
	
	\def\xa{0}
	\def\xb{4}
	
	\def\xe{2}

	\def\xg{2}
	
	\drawEdgesSmallNoGuess
	\drawNodesSmallNoGuess
	\end{scope}
}

\def \drawLBFigSmallGuess{
	\begin{scope}
	\def\yaS{3}
	\def\ybS{2}
	\def\ycS{0}
	
	\def\xa{0}
	\def\xb{4}
	
	\def\xe{2}

	\def\xg{0}	
	\def\xh{1.33}	
	\def\xi{2.66}		
	\def\xj{4}	
	
	\drawEdgesSmallGuess
	\drawNodesSmallGuess
	
	\end{scope}
}

\def \drawNodesSmallest{
	\node at (\xa, \yaS) [nodeSmall] (aa) {};
	\node at (\xb, \ybS) [nodeSmall, fill=white] (ba) {};
	
}

\def \drawEdgesSmallest{
	\drawMultiColorEdgeSmall((\xa, \ybS),(\xa, \yaS))[S][T][.5][0]
}

\def \drawLBFigSmallest{
	\begin{scope}
	
	\def\yaS{1.5}
	\def\ybS{0}
	
	\def\xa{2}
	\def\xb{2}
	
	\drawEdgesSmallest
	\drawNodesSmallest
	
	\end{scope}
}
\tikzstyle{UBnode}=[circle,draw,fill=white, scale=.8, line width=.5mm]
\tikzstyle{edge}=[line width = .15 cm]
\tikzstyle{highlightEdge} = [line width = 1.5mm,densely dashed]

\definecolor{c0}{HTML}{FFFFFF}	
\definecolor{A}{HTML}{4363D8}	
\definecolor{B}{HTML}{00E6E6}	
\definecolor{C}{HTML}{FBBC05}	
\definecolor{D}{HTML}{FF8000}	
\definecolor{E}{HTML}{34A853}	
\definecolor{F}{HTML}{B12DD2}	
\definecolor{G}{HTML}{EA4335}	
\definecolor{H}{HTML}{800000}	

\def\ya{10}
\def\yb{9.5}
\def\yc{9}
\def\yd{8.5}
\def\ye{8}
\def\yf{7.5}

\def\xa{0}

\def\xb{-.5}
\def\xc{.5}

\def\xd{-1}
\def\xe{0}
\def\xf{1}

\def\xg{-1.5}	
\def\xh{-.5}
\def\xi{.5}
	
\def\xj{-2}	
\def\xk{-1}	
\def\xl{0}
\def\xm{1}

\def\xn{-.5}
\def\xo{.5}	
\def\xp{1.5}

\def\drawMulticastNodes[#1]{
	\node at (\xa, \ya) [UBnode,fill=black] (aa) {}; 
	
	\node at (\xb, \yb) [UBnode,fill=#1] (bb) {};
	\node at (\xc, \yb) [UBnode,fill=#1] (cb) {};
	
	\node at (\xd, \yc) [UBnode,fill=#1] (dc) {};
	\node at (\xe, \yc) [UBnode] (ec) {};
	\node at (\xf, \yc) [UBnode,fill=#1] (fc) {};
	
	\node at (\xg, \yd) [UBnode,fill=#1] (gd) {};
	\node at (\xh, \yd) [UBnode] (hd) {};
	\node at (\xi, \yd) [UBnode,fill=#1] (id) {};
	
	\node at (\xj, \ye) [UBnode,fill=#1] (je) {};
	\node at (\xk, \ye) [UBnode] (ke) {};
	\node at (\xl, \ye) [UBnode] (le) {};
	\node at (\xm, \ye) [UBnode,fill=#1] (me) {};
	
	\node at (\xn, \yf) [UBnode] (nf) {};
	\node at (\xo, \yf) [UBnode] (of) {};
	\node at (\xp, \yf) [UBnode] (pf) {};
}

\def\drawMulticastEdges[#1,#2,#3,#4,#5,#6,#7,#8]{
	\draw[edge,color=#1](aa) -- (bb);
	\draw[edge,color=#2](aa) -- (cb);
	
	\draw[edge,color=#1](bb) -- (dc);
	\draw[edge,color=#3](bb) -- (ec);
	\draw[edge,color=#2](cb) -- (fc);
	
	\draw[edge,color=#1](dc) -- (gd);
	\draw[edge,color=#4](dc) -- (hd);
	\draw[edge,color=#2](fc) -- (id);
	
	\draw[edge,color=#1](gd) -- (je);
	\draw[edge,color=#5](gd) -- (ke);
	\draw[edge,color=#6](id) -- (le);
	\draw[edge,color=#2](id) -- (me);
	
	\draw[edge,color=#6](le) -- (nf);
	\draw[edge,color=#7](me) -- (of);
	\draw[edge,color=#8](me) -- (pf);
}

\def\drawSchedulingHighlightsB{
	\draw [highlightEdge](aa) -- node [midway,fill=none,above,sloped,pos =.5] {$P_1$} (bb) --  (dc) -- (gd) -- (je) ;
	\draw [highlightEdge](aa) -- node [midway,fill=none,above,sloped,pos =.5] {$P_2$} (cb) -- (fc)  -- (id) -- (me) ;
}

\def\drawSchedulingHighlightsC{
	\draw [highlightEdge](bb) -- (ec) node [above,sloped,pos=.5] {$P_3$};
	
	\draw [highlightEdge](dc) -- (hd) node [above,sloped,pos=.5] {$P_4$};
	
	\draw [highlightEdge](gd) -- (ke) node [above,sloped,pos=.5] {$P_5$};
	
	\draw [highlightEdge](id) -- (le) node[above, sloped, pos=.5] {$P_6$} -- (nf) ;
	
	\draw [highlightEdge](me) -- (of) node [above,sloped,pos=.5] {$P_7$};
	
	\draw [highlightEdge](me) -- (pf) node [above,sloped,pos=.5] {$P_8$};
}

\def \drawUBGraph{
	\drawMulticastNodes[white]
	\drawMulticastEdges[black,black,black,black,black,black,black,black]
}

\def \drawUBHeavyLight{
	\drawMulticastNodes[white]
	\drawMulticastEdges[A,B,C,D,E,F,G,B]
}

\def \drawUBChopPaths{
	\drawMulticastNodes[white]
	\drawMulticastEdges[A,B,C,D,E,F,G,H]
}

\def \drawUBScheduleA{
	\drawMulticastNodes[white]
	\drawMulticastEdges[A,B,C,D,E,F,G,H]
}

\def \drawUBScheduleB{
	\drawMulticastNodes[white]
	\drawMulticastEdges[A,B,C,D,E,F,G,H]
	\drawSchedulingHighlightsB
}

\def \drawUBScheduleC{
	\drawMulticastNodes[black]
	\drawMulticastEdges[A,B,C,D,E,F,G,H]
	\drawSchedulingHighlightsC
}

\newcommand{\sm}{\text{simultaneous multicast}\xspace}
\newcommand{\sms}{\text{simultaneous multicasts}\xspace}
\newcommand{\SM}{\text{Simultaneous Multicast}\xspace}

\newcommand{\su}{\text{simultaneous unicast}\xspace}
\newcommand{\sus}{\text{simultaneous unicasts}\xspace}

\newcommand{\CONGEST}{\textsc{CONGEST}\xspace}

\begin{document}
\title{Near-Optimal Schedules for Simultaneous Multicasts\thanks{*Supported in part by NSF grants CCF-1527110, CCF-1618280, CCF-1814603, CCF-1910588, NSF CAREER award CCF- 1750808, a Sloan Research Fellowship, funding from the European Research Council (ERC) under the European Union’s Horizon 2020 research and innovation program (ERC grant agreement 949272), Swiss National Foundation (project grant 200021-184735)
        and the Air Force Office of Scientific Research under award number FA9550-20-1-0080. Much of this work was done while the third author was at Carnegie Mellon.} }

\newcommand*\samethanks[1][\value{footnote}]{\footnotemark[#1]}

\author[1]{Bernhard Haeupler}
\author[1]{D Ellis Hershkowitz}
\author[2]{David Wajc}

\affil[1]{Carnegie Mellon University\\
	{\{haeupler, dhershko\}@cs.cmu.edu}}

\affil[2]{Stanford University\\
    {wajc@cs.cmu.edu}}
\date{}
\maketitle

\begin{abstract}

We study the store-and-forward packet routing problem for  simultaneous multicasts, in which multiple packets have to be forwarded along given trees as fast as possible. 

\medskip

This is a natural generalization of the seminal work of Leighton, Maggs and Rao, which solved this problem for unicasts, i.e.\ the case where all trees are paths. They showed the existence of asymptotically optimal $O(C + D)$-length schedules, where the congestion $C$ is the maximum number of packets sent over an edge and the dilation $D$ is the maximum depth of a tree. This improves over the trivial $O(CD)$ length schedules.

\medskip

We prove a lower bound for multicasts, which shows that there do not always exist schedules of non-trivial length, $o(CD)$. On the positive side, we construct $O(C+D+\log^2 n)$-length schedules in any $n$-node network. These schedules are near-optimal, since our lower bound shows that this length cannot be improved to $O(C+D) + o(\log n)$. 


\end{abstract}

\pagenumbering{gobble}

\maketitle
\newpage 

\pagenumbering{arabic}
\section{Introduction}

We study how to efficiently schedule multiple simultaneous multicasts in the store-and-forward model. 

\medskip

Unicasts and multicasts are two of the most basic and important information dissemination primitives in modern communication networks. In a unicast a source sends information to a receiver and in a multicast a source sends information to several receivers. Typically, many such primitives are run simultaneously, causing these primitives to contend for the same resources, most notably the bandwidth of communication links. 

\medskip

The store-and-forward model has been the classic model for developing a clean theoretical understanding of how to most efficiently schedule many such primitives contending for the same link bandwidth. In the store-and-forward model, a network is modeled as a simple undirected graph $G = (V, E)$ with $n$ nodes. Time proceeds in synchronous rounds during which nodes trade packets. In each round a node can send packets it holds to neighbors in $G$, but at most one packet is allowed to be sent along an edge in each round. Nodes can copy packets and send duplicate packets to neighbors, again subject to the constraint that at most one packet crosses an edge each round.\footnote{The assumption that nodes can copy and broadcast packets reflects how standard IP routing works in practice \cite{deering1988host}.}

\medskip

The store-and-forward model, in turn, enables a formal definition of the problem of scheduling many \sms or unicasts. A \sm instance is given by a set of rooted trees $\mcT$---one for each multicast---on a store-and-forward network $G$. Each $T_i \in \mcT$ has root $r_i$ and leaves $L_i$ along with a packet (a.k.a.\ message) $m_i$, initially only known to $r_i$.  
A schedule instructs nodes what packets to send in which rounds, subject to the constraint that $m_i$ can only be sent over edges in $T_i$.
The quality of a schedule is its \emph{length}; i.e., the number of rounds until all nodes in $L_i$ have received $m_i$ for every $i$.
A \su instance is the simple case of a \sm where all $T_i$ are paths. The goal of past work and this work is to understand the length of the shortest schedule.

\medskip


The most important parameters in understanding the length of the shortest schedule has been the \emph{congestion} $C = \max_e |\{T_i \ni e\}|$, i.e., the maximum number of packets that need to be routed over any edge in $G$ and the \emph{dilation} $D = \max_i \text{depth}(T_i)$, i.e., the maximum depth of any multicast-tree or the maximum length of any path in the case of \su. It is easy to see that any schedule requires at least $\max(C, D) = \Omega(C + D)$ rounds: a tree with depth $D$ requires at least $D$ rounds to deliver its message and any edge with congestion $C$ requires at least $C$ rounds to forward all packets that need to be sent over it. On the other hand, any instance can easily be scheduled in $O(CD)$ rounds in a greedy manner: in each round and for each edge $e=(u, v)$, forward $m_i$ from $u$ to $v$ where $T_i$ is an arbitrary tree such that $e \in T_i$ and $u$ knows $m_i$ but $v$ does not; it is easy to verify that this schedule takes $O(CD)$ rounds.


\medskip

Classic results of Leighton, Maggs, and Rao \cite{leighton1994packet} improve upon this trivial $O(CD)$ bound for the case of \su. They showed that  introducing a simple independent random delay for each packet at its source suffices to obtain schedules of length $O(C+D\cdot \log n)$ or $O((C+D)\cdot \frac{\log n}{\log \log n})$. A similar strategy can be shown to also work for simultaneous multicasts \cite{ghaffari2015nearb}. More surprisingly, Leighton et al.~show how an intricate repeated application of the Lov\'asz Local Lemma \cite{alon2016probabilistic} proves the existence of length $O(C + D)$ for any \su instance. This seminal paper initiated a long line of followup work \cite{rothvoss2013simpler,peis2011universal,srinivasan2001constant,bertsimas1999asymptotically,koch2009real,busch2004direct,peis2009packet,rabani1996distributed,ostrovsky1997universal,auf1999shortest,ghaffari2015nearb,meyer1995packet}, some of which even showed these $O(C + D)$-length schedules are efficiently computable \cite{leighton1999fast}, even by \emph{distributed} algorithms \cite{ostrovsky1997universal, rabani1996distributed}.

\medskip

In contrast, essentially nothing beyond the above trivial $O(CD)$ and simple random delay bounds of $O(C+D\cdot \log n)$ and $O((C+D)\cdot \frac{\log n}{\log \log n})$ is known for \sm, despite ample practical and theoretical motivation.  In particular, \sm forms an important component of practical content-delivery systems \cite{jannotti2000overcast,chu2001enabling}, as well as recent theoretical advances in distributed computing \cite{dory2019improved,ghaffari2018new,ghaffari2016distributedI,ghaffari2016distributedII,haeupler2016low,haeupler2016near,haeupler2018minor,haeupler2018round}.

\medskip
The length of the optimal \sm schedule is made all the more intriguing by a recent, award-winning work of \citet{ghaffari2015nearb}. This work studied a natural generalization of \sm, namely how to schedule many simultaneous distributed algorithms, which corresponds to scheduling the routing of messages on directed acyclic graphs (DAGs). Ghaffari showed that in this setting no $O(C+D)$ schedules exist and, in fact, (up to $O(\log \log n)$ factors) the random delay upper bound of $O(C+D\cdot \log n)$ is the closest that one can get to an $O(C+D)$ bound. Given that multicasts are more general than unicasts but less general than DAGs, it has remained an interesting open question whether an $O(C + D)$ schedule comparable to those for unicasts is also possible for multicasts or whether, like for DAGs, a multiplicative $O(\log n)$ overhead is required. 

\subsection{Our Contributions}

We show that, unlike in the unicast setting where $O(C + D)$ schedules are possible, for multicasts the trivial $O(CD)$ upper bound cannot be improved without introducing a dependence on the number of nodes, $n$.

\begin{restatable}{thm}{lowerBound}\label{lem:LB}
For any $C, D, n \in \mathbb{Z}_+$ such that $C 2^{D + 1} \leq \log n$ there exists a simultaneous multicast instance on an $n$-node graph with congestion $C$ and dilation $D$ whose optimal schedule requires at least $\tfrac{CD}{2}$ rounds.
\end{restatable}
\noindent We note that our lower bound also implies a new lower bound of $\Omega(CD)$ for the DAGs case studied by \citet{ghaffari2015nearb} since the DAGs case generalizes \sms.

On the positive side, we show that if one allows a schedule's length to depend on $n$ then, unlike in the DAGs case where $O(C + D \cdot \log n)$ is the closest one can get to $O(C + D)$, one can get $O(C + D)$ with a mere \emph{additive} $O(\log ^ 2 n)$.

\begin{restatable}{thm}{upperBoundCentralized}
	\label{thm:UBCentralized} 
Each \sm instance with congestion $C$ and dilation $D$ in an $n$-node network admits a schedule of length at most $O(C + D + \log^2 n)$.
\end{restatable}
\noindent We also verify that these schedules are efficiently computable both by a deterministic, centralized polynomial-time algorithm and by a randomized distributed algorithm in the \CONGEST model. 

Complementing our proof that shows the existence of $O(C + D + \log^2 n)$ schedules, we extend our lower bound to show that any schedule with purely additive dependence on $C$, $D$ and any function of $n$ incurs at least an additive $\Omega(\log n)$ term. This implies that the additive $\log^2 n$ in \Cref{thm:UBCentralized} is essentially optimal.

\begin{restatable}{thm}{lowerBoundAdd}\label{lem:LBadd} 
Suppose there is a function $f$ such that for any \sm instance with congestion $C$ and dilation $D$, there is a schedule delivering all packets in $O(C + D) + f(n)$ steps. Then $f(n) = \Omega(\log n)$.
\end{restatable}
\noindent In summary, our results give an essentially optimal characterization of what \sm schedules are possible and cleanly separate the complexity of \sm schedules from those of simultaneous unicasts and DAGs.

\subsection{Related Work}\label{sec:related-word}
We will take this section to summarize additional related work.

\medskip 
\noindent\textbf{Existence of Good Simultaneous Unicast Schedules.}
The seminal work of \citet{leighton1994packet} initiated a series of works aimed at showing short \su schedules exist. For example, \cite{scheideler2006universal,peis2011universal} improved the constants in the $O(C + D)$ schedules of Leighton et al., with \cite{peis2011universal} also generalizing this result to edges with non-unit transit times and bandwidth. \citet{rothvoss2013simpler} presented a simplified proof compared to that of \cite{leighton1994packet} by way of the ``method of conditional expectations'', and also increased the constant in the $\Omega(C+D)$ lower bound.

%

\medskip
\noindent\textbf{Scheduling Other Simultaneous Algorithms.}
In addition to the mentioned work of \citet{ghaffari2015nearb}, there is a variety of work in scheduling of specific distributed algorithms. A classic result of \citet{topkis1985concurrent} shows that $h$-hop broadcast of $k$ messages from different sources can be done in $O(k + h)$ rounds. 
This is a special case of \sm, where $k$ multicast instances are to be scheduled along edges of trees with congestion $C\leq k$ and depth $D$. So, for this special case of \sm a $O(C+D)$-length schedule always exists.
More recently, \citet{holzer2012optimal} showed that $n$ BFSs can be performed from different nodes in $O(n)$ rounds. This was generalized by \citet{lenzen2013efficient} who showed that $k$ many $h$-hop BFSs from different sources can be done in $O(k + h)$ rounds. 


\medskip
\noindent\textbf{Algorithmic Results.}
Another line of work on \su and related problem focused on \emph{computing} optimal or near-optimal schedules efficiently, starting with work of \citet{leighton1999fast}.
There has been work on \su focused on ``local-control'' or distributed algorithms, where at each step each node makes decisions on which packets to move forward
along their paths, based only on the routing information that
the packets carry and on the local history of execution.
The $O(C+D\cdot \log n)$ algorithm of \citet{leighton1994packet}, for example, is such a distributed \su algorithm.
\citet{rabani1996distributed} improved this bound to $O(C) + D\cdot (\log^*n)^{O(\log ^* n)}$ rounds, which was then further improved by \citet{ostrovsky1997universal} to $O(C + D + \log^{1+\epsilon} n)$ rounds for any constant $\epsilon > 0$. 
Another series of works also studied centralized algorithms for simultaneous unicast where the source and sink pairs are fixed but the algorithm is free to choose what paths it uses to deliver packets from sources to sinks. Notably, \citet{srinivasan2001constant} gave a constant approximation for this problem. \citet{bertsimas1999asymptotically} then provided an asymptotically-optimal algorithm, outputting a schedule of length $OPT+(\sqrt{n\cdot OPT})$; i.e., $OPT(1+o(1))$ for sufficiently large $OPT$. Lastly, there has been work in computing schedules for single multicasts \cite{bar2000message,elkin2003sublogarithmic} and even \sms \cite{iglesias2015rumors,iglesias2018plane} in models fundamentally different from the store-and-forward model we study.

\section{Intuition and Overview of Techniques}\label{sec:intuition}
We now give an overview of and intuition for the techniques we use in our main results.

\subsection{$\Omega(CD)$ Lower Bound}\label{sec:lowerInt}
The goal of our lower $\Omega(CD)$ lower bound construction is to repeatedly ``accumulate'' delays by combining together already delayed multicast trees. Here, we build some intuition for this strategy . 

Consider a \sm instance consisting of two trees $S$ and $T$ using a single edge as in \Cref{fig:LBSmallest}. Since at most one message crosses this edge each round, we know that after one round at least one of our trees' messages will be delayed by $1$ round, i.e., will not have crossed the edge.
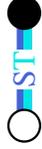
\begin{figure}
	\centering
	\begin{tikzpicture}[scale=1]
	\drawLBFigSmallest
	\end{tikzpicture}\caption{A congested edge example on multicast trees $S$ and $T$. Root of both trees given by black node. Each multicast tree given in a different color and edges labeled by which multicast trees use them.}\label{fig:LBSmallest}
\end{figure}
\noindent More generally, if $C$ trees all use a single edge $e$ then for any fixed schedule one of these trees will require at least $C$ rounds until its message crosses $e$. 

If we knew, a priori, for any $C$-congested edge which multicast tree was delayed by $C$, producing a hard multicast instance would be easy as we could repeatedly combine together the multicast trees delayed by $C$ in each congested edge. For instance, consider the following example, illustrated in \Cref{subfig:noGuess} where $C=2$.  We have four multicast trees $S$, $T$, $U$ and $V$ where $S$ and $T$ have root $r_1$ and $U$ and $V$ have root $r_2$. Both roots connect to a vertex $v$ where $(r_1, v)$ is used by $S$ and $T$ and $(r_2, v)$ is used by $U$ and $V$. If we knew that after a single round $T$ and $V$ used edges $(r_1, v)$ and $(r_2, v)$ respectively, then we could ``combine'' $S$ and $U$ into a new edge $(v, u)$. Then, the messages for $S$ and $U$ wouldn't arrive at $v$ until at least two rounds have passed and since both $S$ and $U$ use the edge $(v, u)$, one of the messages of either $S$ or $U$ wouldn't arrive at $u$ until four rounds have passed, despite the fact that $u$ is only two hops from the root of each tree. We might hope, then, to recursively repeat this strategy, combining together such gadgets to accumulate larger and larger delays. 

However, we, of course, do not always know which trees are delayed and so combining together the most delayed tree is not a feasible strategy. That is, we must provide a construction which requires many rounds for \emph{every possible} \sm schedule, not many rounds for one fixed schedule.


We overcome this challenge by using the fact that trees, unlike paths, branch. In particular, we will use the branching of trees to ``guess'' which tree was delayed for every congested edge. 
As a concrete example of this strategy consider the \sm instance given in \Cref{subfig:guess}. We have the instance as in \Cref{subfig:noGuess} but now instead of vertex $u$, we have four vertices, one for each possible guess of which pair of elements in $\{S, T\} \bigtimes \{U, V\}$ are delayed at $(r_1, v)$ and $(r_2, v)$. Now notice that for any fixed \sm schedule for this instance we know that after one round only one of $S$ and $T$'s messages will cross $(r_1, v)$ and only one of $U$ and $V$'s message will cross $(r_2, v)$. Without loss of generality suppose $S$ and $U$ do not cross $(r_1, v)$ and $(r_2, v)$ respectively in the first round. We then know that one of the edges $(v, u')$ corresponding to one of our guesses---in this case the edge used by $S$ and $U$---is such that the trees which use this edge will not deliver the their messages to $v$ until two rounds have passed. Similarly, we know that at most one of $S$ and $U$'s messages arrive at $u'$ by the third round---without loss of generality $U$'s message. Thus, $S$ will not successfully deliver its message to all leaves until at least four rounds have passed, despite the fact that all leaves of $S$ are only two hops from $S$'s root. More generally, if we repeated this construction with a larger congestion $C$ we would have that some multicast tree requires at least $2C$ rounds to deliver its message to all leaves, despite the fact that $C + D = C + 2$.
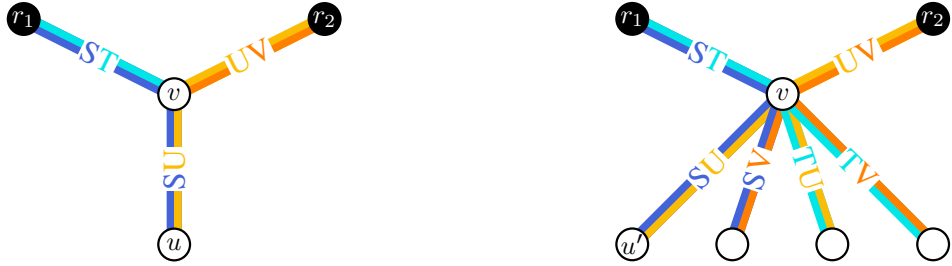
\begin{figure}
	\centering
	\begin{subfigure}[t]{0.49\textwidth}
		\centering
		\begin{tikzpicture}[scale=1]
		\drawLBFigSmallNoGuess
		\end{tikzpicture}
		\caption{Accumulating delays if we know $S$ and $U$ delayed.}\label{subfig:noGuess}
	\end{subfigure}%
	\begin{subfigure}[t]{0.49\textwidth}
		\centering
		\begin{tikzpicture}[scale=1]
		\drawLBFigSmallGuess
		\end{tikzpicture}
		\caption{Guessing the delayed trees.}\label{subfig:guess}
	\end{subfigure}\caption{Illustration of how one can ``guess'' which trees are delayed. Roots given by black nodes. Each multicast tree given in a different color and edges labeled by which multicast trees use them.}\label{fig:LBSmall}
\end{figure}

Our lower bound construction will recursively stack trees like those in \Cref{subfig:guess} to guess which multicast trees a schedule chooses to delay and accumulate a larger and larger delay by combining together these delayed trees. We will guarantee that some sub-graph is always correct in its guesses. We will also make use of the observation that if $C$ trees all use a single edge $e$ then by Markov's inequality at least $\tfrac{C}{2}$ of these trees will require $\frac{C}{2}$ rounds until their message crosses $e$ to reduce the amount of guessing we must do; this will allow us to expand the possible values of $C$ and $D$ we can use when constructing our lower bound graph which will aid in the proof of \Cref{lem:LBadd}.
We elaborate on our $\Omega(CD)$ construction in \Cref{sec:lowerBound} and then extend it in \Cref{sec:AddLogn} to show that additive $\Omega(\log n)$ is necessary for length $O(C + D)$ schedules.

\subsection{Existence of $O(C + D + \log^2 n)$-Length \SM Schedules}
The main intuition underlying our $O(C + D + \log^2 n)$-length \sm schedules is that every instance of multicast can be reduced to a series of unicast instances and, as \citet{leighton1994packet} showed, unicast instances admit schedules of length linear in their congestion and dilation. Our goal then is to gracefully reduce a \sm to a series of \sus.

Here, we discuss two natural approaches for such a reduction, argue that they fail and extract intuition for our upper bound from this failure. In the first approach, for each multicast tree $T_i$ we define $|L_i|$ unicast instances, where for each leaf $l \in L_i$ we have a unicast instance on the root-to-leaf path from $r_i$ to $l$. While this \su instance has dilation $D'=D$, it also has congestion potentially as high as $C'=\Omega(n)$: unicasts corresponding to the same tree are run independently, and each edge in $T_i$ is contained in every root-to-leaf unicast path. Relying on the existence of schedules of length $O(C'+D')$ guaranteed by \cite{leighton1994packet}, then, could yield schedules of length as bad as $\Omega(D + n)$. In the second approach, we define a separate unicast instance for each edge in each $T_i$. We then run a \su schedule for all edges from roots of multicast trees to their children, then from roots' children to their children, and so on and so forth. Here we have at least obtained a sequence of \su instances with lower dilation---$D'=1$---and congestion no larger than what we started with---$C'\leq C$. \citet{leighton1994packet} guarantees the existence of schedules of length $O(C'+D')$ for each such \su instance. Unfortunately, we must concatenate together the schedules of $D$ such \su instances to solve the \sm instance, which would yield schedules of length $\Omega(D(C'+D'))=\Omega(CD)$; i.e., no better than the trivial schedule.

Thus, the challenge in reducing \sm to \su is finding a suitable way of balancing between these two approaches. In the first reduction, we were able to solve a single \su problem with dilation $D$ but one whose congestion was much larger than the congestion of the \sm problem with which we started. In the second extreme, we were able to solve \su instances with dilation and congestion only $1$ and $C$ but we had to solve many such problems. 

Our goal, then, is to find a way of reducing \sm to \su in a way that keeps the dilation and congestion of the resulting \su instances small but does not require solving too many \sus. We strike such a balance by computing what we call a \emph{$(\log n,\log n)$-short path decomposition} of each multicast tree. This decomposition is based on subdividing paths in the heavy path decompositions of \citet{sleator1983data}. By using such a decomposition on each multicast tree along with random delays determining when to schedule each path in the decomposition, we obtain a sequence of $\frac{C}{\log n}+\frac{D}{\log n}+\log n$ many \su instance whose congestion $C'$ and dilation $D'$ are both at most $O(\log n)$ with high probability. Relying on the $O(C'+D')$ schedules guaranteed by \citet{leighton1994packet} for these \su instances, we find that every \sm instance admits a  schedules of length $O(C + D + \log^2 n)$. We elaborate on this in \Cref{sec:upper-bound}. We also provide centralized and distributed algorithms for the computation of these schedules in \Cref{sec:algorithms}.

\section{$\Omega(CD)$ Lower Bound}\label{sec:lowerBound}

This section is dedicated to the proof of our $\Omega(CD)$ lower bound. We begin this section by providing the family of instances we use to show this lower bound. We proceed to show how this family requires $\Omega(CD)$ rounds, showing that $O(C+D)$ \sm schedules are generally impossible and that the trivial $O(CD)$ schedule is the best \sm schedule without a dependence on $n$. Specifically, we prove the following.
\lowerBound*

\subsection{Multicast Instance}\label{sec:multiLBInstance}
We will describe how our instance is constructed in a top-down manner. For the remainder of this section we fix a desired congestion $C$ and dilation $D$. We will recursively construct a graph in which every edge receives $C$ ``labels'' where the graph induced by each label is a distinct multicast tree.\footnote{The graph induced by a label $\chi$ on graph $G$ is the subgraph of edges from $G$ labeled $\chi$.}  As each label corresponds to a multicast tree, each label will also have a root corresponding to it which will be the root of the corresponding multicast tree. Ultimately, our instance corresponding to a fixed $C$ and $D$ will contain $C \cdot 2^{D-1}$ multicast trees and so throughout this section we will imagine we have $C \cdot 2^{D-1}$ distinct labels. Throughout this section we will also let capital letters correspond to labels; e.g.\ $\{S,T,U,V,W,X,Y,Z\}$ is a set of $8$ labels. Before moving onto specific details, we refer the reader to \Cref{fig:LBStepThree} for a visual preview of our lower bound construction.


\subsubsection{Interleaving Labels}

In order to rigorously define what it means to guess which multicast trees are delayed, we introduce the idea of ``interleaving'' the sets of labels corresponding to our multicast trees.

Given sets $S_1$ and $S_2$, each consisting of $C$ labels, we let the interleaving of $S_1$ and $S_2$ be $I(S_1, S_2) := \{S'_1 \cup S_2' : S_i' \subseteq S_i, |S'_i| = C/2 \}$ be all subsets which take $C/2$ labels from $S_1$ and $C/2$ labels from $S_2$. For example, if $C=2$ and $S_1 = \{S, T\}$ and $S_2 = \{U, V\}$ then $I(S_1, S_2) = \{\{S,U\}, \{S,V\}, \{T, U\} ,\{T,V\} \}$. $S_1$ and $S_2$ will correspond to two adjacent edges, each in a disjoint set of $C$ multicast trees each and so $I(S_1, S_2)$ will correspond to all ways of guessing which $C$ trees, taking $C/2$ trees from one edge and $C/2$ trees from the other edge, are delayed among the $2C$ multicast trees which use one of the two edges.

Let $\mcS = (S_i)_{i=1}^{2^{D-1}}$ be a tuple partitioning our $C \cdot 2^{D-1}$ distinct labels into sets of size $C$. That is, each $S_i$ is a set (with associated index $i$) containing $C$ distinct labels and $S_i \cap S_j = \emptyset$ for $i \neq j$. We call two sets in $\mcS$ adjacent if the index of one is $2i-1$ and the index of the other is $2i$ for some $i \in \mathbb{Z}_{\geq 1}$. Finally, we let
\begin{align*}
I(\mcS) := \bigtimes_{i=1}^{|\mcS|/2} I(S_{2i}, S_{2i + 1})
\end{align*}
be all possible interleavings of adjacent sets in $\mcS$ where $\bigtimes$ denotes an $|S|/2$-wise Cartesian product. This tuple $\mcS$ will correspond to all edges of our construction at height $D$ while a pair of adjacent sets $S_{2i-1}$ and $S_{2i}$ will correspond to two adjacent edges, each in disjoint sets of $C$ multicast trees; thus $\mcI(\mcS)$ will correspond to all possible ways to guess how, among all pairs of adjacent edges at height $D$, which $C$ trees, taking $C/2$ trees from one edge and $C/2$ trees from the other edge, are delayed in each pair.

We give a concrete example of our notation where $C = 2$ and $D = 3$. Let $\mcS = (S_1, S_2, S_3, S_4) = (\{S, T\}, \{U, V\}, \{W, X\}, \{Y, Z\})$.  Then $I(\mcS)$ corresponds to all ways of combining $S_1$ and $S_2$ by taking one element from each and all ways of combining $S_3$ and $S_4$ by taking one element from each. In particular, we have 
\begin{align*}
I(\mcS) &= I(S_1, S_2) \times I(S_3, S_4)\\
&= \big\{\{S,U\}, \{S, V\}, \{T,U\}, \{T, V\} \big\} \times \big\{\{W,Y\}, \{W, Z\}, \{X,Y\}, \{X, Z\} \big \},
\end{align*}
which is
\begin{align*}
&= \big\{ (\{S,U\}, \{W,Y\}) , &(\{S,U\}, \{W,Z\}), && (\{S,U\}, \{X,Y\}), && (\{S,U\}, \{X,Z\}),\\
&~~~~~~  (\{S,V\}, \{W,Y\}) , & (\{S,V\}, \{W,Z\}), && (\{S,V\}, \{X,Y\}), && (\{S,V\}, \{X,Z\}),\\
&~~~~~~  (\{T,U\}, \{W,Y\}) , & (\{T,U\}, \{W,Z\}), && (\{T,U\}, \{X,Y\}), && (\{T,U\}, \{X,Z\}),\\
&~~~~~~  (\{T,V\}, \{W,Y\}) , & (\{T,V\}, \{W,Z\}), && (\{T,V\}, \{X,Y\}), && (\{T,V\}, \{X,Z\}) \big\}.
\end{align*}
Notice that each $\mcS' \in I(\mcS)$ is a tuple of sets, each of $C$ labels, and the number of distinct labels across all sets in $\mcS'$ is exactly half of the number of labels across all sets in $\mcS$. Each $\mcS'\in I(\mcS)$ will correspond to a single recursive call in our construction.

\subsubsection{Our Instance}

With the above notation in hand, we now describe how we recursively construct our multicast instance for a fixed $C$ and $D$. Let $\mcS = (S_i)_{i=1}^{2^{D-1}}$ be an arbitrary partition of our $C \cdot 2^{D-1}$ distinct labels into sets of size $C$ as above. Our recursion will be on $D$; that is, we will recursively construct several instances of \sm with dilation $D-1$ and congestion $C$ and then combine together these instances into a single instance with congestion $C$ and dilation $D$ (in fact, every edge will have congestion exactly $C$ and every tree will have depth exactly $D$). We let $M_\mcS$ be the \sm instance we construct from $\mcS$ and let $G_\mcS$ be its corresponding graph. As can easily be seen by induction on our recursive depth, each root will inductively only be incident on a single edge and so we let $\chi(r)$ be the set of labels of the one edge incident to root $r$.

\begin{itemize}
	\item \textbf{Base case ($D=1$):}  In this case we have $\mcS = (S_1)$. We let $M_\mcS$ consist of one edge $(r, v)$ which receives every label in $S_1$ and let $r$ be the root of every label/tree.
	\item \textbf{Inductive case ($D > 1$):} We construct $M_\mcS$ in three steps. See Figures \ref{fig:LBStepOne}, \ref{fig:LBStepTwo} and \ref{fig:LBStepThree} for an illustration of the results of steps one, two and three respectively. In the example in these figures $C = 2$, $D=3$ and $\mcS = (S_1, S_2, S_3, S_4) = (\{S,T\}, \{U, V\}, \{W,X\}, \{Y,Z\})$; roots of multicast trees are indicated by black nodes and edges are colored according to their label/tree.
	\begin{enumerate}
		\item First, for each pair of adjacent sets $S_{2i-1}$ and $S_{2i}$ in $\mcS$ we introduce vertices $r_{2i-1}$, $r_{2i}$ and $v_i$ and edges $e_{2i-1} = (r_{2i-1}, v_i)$ and $e_{2i} = (r_{2i}, v_i)$. We let $r_j$ be the root of all trees with a label in $S_j$ and $e_j$ receive all labels in $S_j$ for every $j$ ( \Cref{fig:LBStepOne}).
		\item Next, we ``guess'' which trees will be delayed on the $e_j$s. In particular, we add to our graph the disjoint union of $G_{\mcS'}$ for each $\mcS' \in \mcI(\mcS)$; each of the new connected components in our graph at this point corresponds to some instance $M_{\mcS'}$; each edge inherits the $C$ labels it received in $M_{\mcS'}$ (\Cref{fig:LBStepTwo}).
		\item Finally, we connect up our guesses to the corresponding parents. In particular, for each vertex $r$ that was a root in $M_\mcS'$ if $\chi(r) \in I(S_{2i-1}$, $S_{2i})$, we identify $r$ and $v_i$ as the same vertex (\Cref{fig:LBStepThree}).
	\end{enumerate}

\end{itemize}

It is easy to verify by induction on our recursive depth that, indeed, each label and its root induce a tree in the returned graph and so $M_\mcS$ is an instance of \sm; see \Cref{app:lowerBound} for a proof.

\begin{restatable}{lem}{lowerBoundIsMulticast}
	\label{lem:isMulticastInstance} 
	Let $\mcS = (S_i)_i^{2^{D-1}}$ be a partition of $C \cdot 2^{D-1}$ distinct labels as above. Then each label in $\bigcup_i S_i$ induces a rooted tree in $G_\mcS$.
\end{restatable}



\begin{figure}
	\centering
	\begin{tikzpicture}[scale=1]
	\drawLBConstructStepOne
	\end{tikzpicture}\caption{The result of step one of our lower bound construction. Notice that we have an edge for each set $S_i$ and each pair of adjacent $S_i$ are joined together at a vertex. We outline in red and green $v_1$ and $v_2$ respectively. Left-to-right black nodes are $r_1$, $r_2$, $r_3$ and $r_4$.}\label{fig:LBStepOne}
\end{figure}

\begin{figure}
	\centering
	\begin{tikzpicture}[scale=1]
	\drawLBConstructStepTwo
	\end{tikzpicture}\caption{The result of step two of our construction. Notice that we now have a new connected component for each $\mcS' \in I(\mcS)$, each of which corresponds to a guess for which trees will be delayed at $v_1$ and $v_2$. We outline in red and green the vertices which in step three we identify with $v_1$ and $v_2$ respectively.}\label{fig:LBStepTwo}
\end{figure}

\begin{figure}
	\centering
	\begin{tikzpicture}[scale=1]
	\drawLBConstructStepThree
	\end{tikzpicture}\caption{Our construction (i.e.\ the result of step three) for $\mcS=(S_1, S_2, S_3, S_4) = (\{S,T\}, \{U, V\}, \{W,X\}, \{Y,Z\})$, $C=2$ and $D=3$. Notice that we have modified the graph in \Cref{fig:LBStepTwo} by adding one edge for each $S_{2i-1}$ (resp. $S_{2i}$) colored by the labels in $S_{2i-1}$  (resp. $S_{2i}$) going from root $r_{2i-1}$  (resp. $r_{2i}$) to $v_i$.}\label{fig:LBStepThree}
\end{figure}
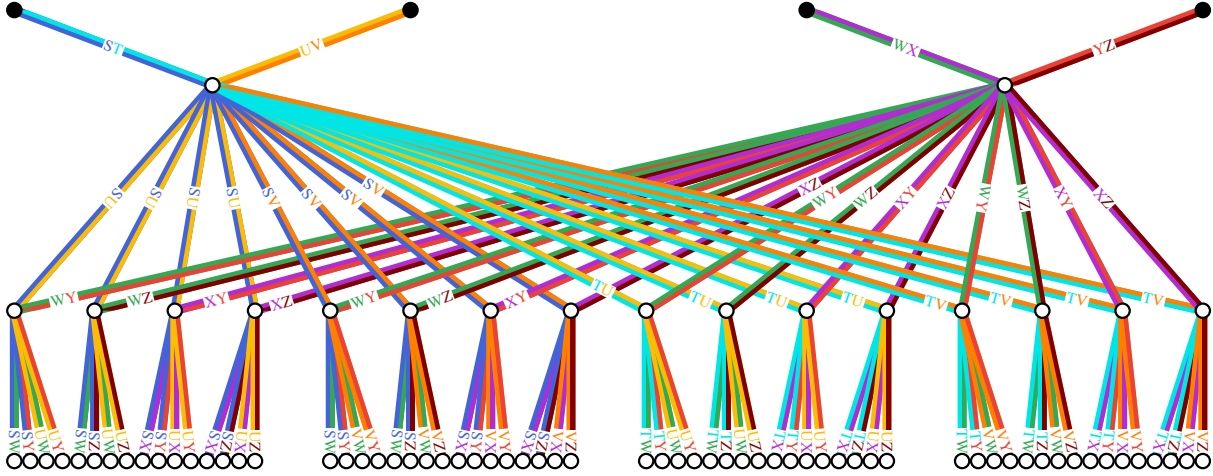

\subsection{Proof of $\Omega(CD)$ Lower Bound}
An induction on $D$ demonstrates that our \sm instance has the appropriate congestion and dilation. We defer a proof of this simple lemma to \Cref{app:lowerBound}.

\begin{restatable}{lem}{correctCAndD}
	\label{lem:conAndDil} 
	$M_\mcS$ has congestion $C$ and dilation $D$.
\end{restatable}

Another simple induction on $D$ and standard approximations allows us to bound the number of nodes in our lower bound graph;  see \Cref{app:lowerBound} for a proof.
\begin{restatable}{lem}{nUpperBound}
	\label{lem:nBound} 
	$|V(G_\mcS)| \leq 2^{C(2^{D + 1})}$.
\end{restatable}

Having established the basic properties of our instance, we now argue that (asymptotically) the best one can hope for on our instance is the trivial $O(CD)$-round schedule. As discussed in \Cref{sec:lowerInt}, we will prove this by arguing that for any fixed schedule, some sub-graph in $G$ was correct in ``guessing'' which multicast trees were slowed down. In particular, we will argue that for any fixed schedule there is some smaller instance of \sm which this schedule must solve as a sub-problem which takes at least $\frac{C(D-1)}{2}$ rounds but which the schedule does not start making progress towards solving until at least $\frac{C}{2}$ rounds have passed.
\begin{lem}\label{lem:optLB}
	The optimal schedule on $M_\mcS$ is of length at least $\frac{C D}{2}$.
\end{lem}
\begin{proof}
	Fix an arbitrary \sm schedule. We will prove by induction on $D$ that $M_{\mcS}$ requires at least $\frac{CD}{2}$ rounds. The base case of $D=1$ is trivial, as in this case $M_{\mcS}$ is a single edge with congestion $C$ and so clearly requires at least $C\geq \frac{CD}{2}$ rounds.
	
	For the inductive step, $D>1$, suppose that for any partition $\mcS'$ of $C \cdot 2^{D-1}$ distinct labels into sets of size $C$, we have that $M_{\mcS'}$ requires at least $\frac{C(D-1)}{2}$ rounds. By definition of $M_{\mcS}$, any schedule which solves $M_{\mcS}$ can be projected in the natural way onto $M_{\mcS'}$ as a schedule which solves $M_{\mcS'}$ for any $\mcS' \in I(\mcS)$. For example, any schedule which solves the instance in \Cref{fig:LBStepThree} induces a schedule which when projected onto \Cref{fig:LBStepTwo} solves $M_{\mcS'}$ for each of the recursively constructed $M_{\mcS'}$. Even stronger, notice that $M_{\mcS}$ is created by combining the union of all $M_{\mcS'}$ for $\mcS' \in I(\mcS)$ in such a way that any schedule which solves $M_{\mcS}$ must also send all messages from roots of trees in $M_\mcS$ to corresponding roots of trees in $M_{\mcS'}$ and solve $M_{\mcS'}$. That is, let $r'$ be an arbitrary root for $M_{\mcS'}$ and let $\chi(r')$ be the labels associated with the one edge for root $r'$ in $M_{\mcS'}$. Then, if we identify $r'$ with $v_i$ when constructing $M_\mcS$ then a schedule for $M_{\mcS}$ must both send $r'$ the $C$ messages of $\chi(r')$ from $r_{2i-1}$ and $r_{2i}$ and solve $M_{\mcS'}$. Thus, clearly the time our schedule takes is at least the time it takes to send one message in $\chi(r')$ to $r'$ from $r_{2i-1}$ and $r_{2i}$ for some root in $M_{\mcS'}$ plus the time it takes to solve $M_{\mcS'}$. For example, if we let $\mcS' = (\{S,U\}, \{W,Y\})$ then any schedule which solves \Cref{fig:LBStepThree} solves $M_{\mcS'}$ but before doing so must clearly send at least one message of $\{S,U\}$ to $v_1$ or at least one message  $\{W,Y\}$ to $v_2$.
	
	We now define $\bar{\mcS}'$ where $\bar{\mcS}' \in I(\mcS)$ so that $M_{\bar{\mcS}'}$ is an instance of \sm embedded in $M_{\mcS}$ which our fixed schedule must solve in order to solve $M_{\mcS}$ but which it does not start start solving until at least $C/2$ rounds have passed. In particular, consider what our fixed schedule does in the first $C/2$ rounds. As $C$ messages must cross each edge $e_j$ and only one such message can cross per rounds, there are some $C/2$ multicast trees whose message cannot cross $e_j$ before $C/2$ rounds have passed; let $S'_j$ be these ``slow'' trees for edge $e_j$ and let 
	\begin{align*}
	\bar{\mcS}' := (S_{2i-1}' \cup S_{2i}')_{i=1}^{2^{D-2}}
	\end{align*}
	be a partition of the labels corresponding to these slow edges into sets of size $C$.
	
	The tuple $\bar{S}'$ belongs to $I(\mcS)$ and so, as discussed above, the fixed schedule must first send at least one message to a root in $M_{\bar{S}'}$ and then solve $M_{\bar{S}'}$. By definition of $\bar{S}'$, no messages arrive at roots of trees in $M_{\bar{S}'}$ until at least $\frac{C}{2}$ rounds have passed. On the other hand, by the inductive hypothesis, the latter sub-instance takes at least $\frac{C(D-1)}{2}$ additional rounds. Thus, the schedule must use at least $\frac{C(D-1)}{2} + \frac{C}{2} = \frac{CD}{2}$ rounds.
\end{proof}

Combining Lemmas \ref{lem:conAndDil}, \Cref{lem:nBound} and \ref{lem:optLB} and noting that we can always add dummy nodes to increase the number of vertices in our graph to a desired $n$ immediately yields \Cref{lem:LB}.

\section{Existence of $O(C + D + \log^2n)$-Length Schedules}\label{sec:upper-bound}

Here we demonstrate that length $O(C + D + \log^2n)$ \sm schedules always exist.
For this result we rely on heavy path decompositions, first introduced by \citet{sleator1983data}.

\begin{Def}[Heavy path decomposition \cite{sleator1983data}]
	A \emph{heavy path decomposition} of a rooted tree $T$ is obtained as follows. First, each non-leaf node selects one \emph{heavy} edge, which is an edge to a child with the greatest number of descendants (breaking ties arbitrarily). Other edges are termed \emph{light}. 
	We consider inclusion-wise maximal paths consisting of heavy edges, and for each highest node $v$ of such a path $p$, we add to the path $p$ the edge from $v$ to its parent (if any). The obtained paths form the heavy path decomposition.
\end{Def}

It is easy to see that this is indeed a decomposition of the tree; that is, that each edge belongs to exactly one path in the heavy path decomposition. Moreover, each root-to-leaf path intersects at most $\log_2n$ heavy paths, as each such path can have at most $\log_2n$ light edges because the number of nodes in a subtree decreases by at least a factor of two every time one traverses down a light edge. This will allow us to decompose the trees into ``short paths'' such that each root-to-leaf path intersects few short paths. 
Specifically, we define a refinement of this decomposition in a top-down fashion, by breaking up each heavy path into short paths of length at most $\log_2 n$; that is, starting from the top of a heavy path of length $l$, we cut it into $\lceil l / \log n\rceil$ short paths. See \Cref{fig:UBGraphDecomp}.
Both the decomposition and its refinement exist, and are even computable deterministically in linear time.

\begin{figure}
	\centering
	\begin{subfigure}[t]{0.32\textwidth}
		\centering
		\begin{tikzpicture}[scale=1]
		\drawUBGraph
		\end{tikzpicture}
		\caption{Multicast tree $T_i$.}
	\end{subfigure}%
	\begin{subfigure}[t]{0.32\textwidth}
		\centering
		\begin{tikzpicture}[scale=1]
		\drawUBHeavyLight
		\end{tikzpicture}
		\caption{Heavy-light decomposition.}\label{subfig:hlDecomp}
	\end{subfigure}%
	\begin{subfigure}[t]{0.32\textwidth}
		\centering
		\begin{tikzpicture}[scale=1]
		\drawUBChopPaths
		\end{tikzpicture}
		\caption{Cut heavy paths into short paths.}\label{subfig:decompCutPaths}
	\end{subfigure}\caption{Our decomposition for multicast tree $T_i$. Each heavy path in \Cref{subfig:hlDecomp} and short path in \Cref{subfig:decompCutPaths} drawn in different colors. Notice the far right path is cut into two short paths since its length is $5 > \log_2 n =4$.}\label{fig:UBGraphDecomp}
\end{figure}
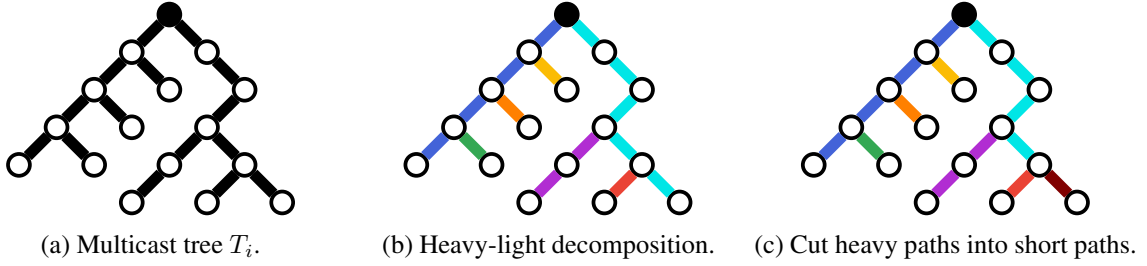

As each root-to-leaf path intersects at most $\log_2n$ heavy paths, this refined decomposition has each root-to-leaf path intersect at most $\frac{D}{\log_2 n} + \log_2 n$ short paths. We will refer to such a decomposition as a $(\log n,\log n)$-short (path) decomposition. We use this particular name as we generalize this notion further in \Cref{sec:algorithms} to $(l, k)$-decompositions for any integers $k$ and $l$.
This refined path decomposition together with some additional random delays will allow us to reduce the task of \sm to that of $O\big(\frac{C}{\log n}+\frac{D}{\log n}+\log n\big)$ many \su instances with congestion and dilation $O(\log n)$, from which we obtain the following result. We illustrate the schedules in this result in \Cref{fig:UBSchedule}.

\upperBoundCentralized*
\begin{proof}
	We prove this result by means of the probabilistic method. 
	First, we consider a $(\log n,\log n)$ decomposition of each multicast tree. 
	For each short path $p$ in the $(\log n,\log n)$-decomposition of a tree, we say $p$ is at \emph{level} $j$ if there are exactly $j-1$ other short paths between $p$'s root and the tree's root. That is, if we were to schedule a particular tree by forwarding along all paths of level $j=1,2,\dots$ during $R=O(\log n)$ rounds, the path of level $j$ would be scheduled in rounds $((j-1)\cdot R, j\cdot R]$, which we refer to as the $j$-th \emph{frame}.
	Our goal will be to schedule the sets of short paths with limited congestion in parallel, using \su schedules guaranteed by \citet{leighton1994packet}. 
	
	In order to break up \sm to multiple \sus, we shift the levels of each tree $T_i$ by a random offset $X_{T_i}$ chosen uniformly in $[C/\log n]$. Now a short path of level $j$ in tree $T_i$ will be scheduled during frame $j+X_{T_i}$.
	Since each edge $e$ has congestion $C$, the expected number of paths of different trees that use $e$ during any given frame is at most $O(\log n)$. So, by standard Chernoff concentration inequalities, the congestion of each edge during any frame is at most $O(\log n)$ w.h.p.  Therefore, applying a union bound over all edges and time frames, we find that w.h.p.,~all edges have congestion at most $O(\log n)$ for all (shifted) frames $j=1,2,\dots,\frac{C}{\log n}+\frac{D}{\log n}+\log n$ (recall that each root-to-leaf path intersects at most $\frac{D}{\log n}+\log n$ paths of length at most $\log n$). 
	In particular, there exist random delays such that each time frame consists of a \su instance with congestion $C'=O(\log n)$ and dilation $D'=O(\log n)$. Therefore, by  \citet{leighton1994packet,leighton1999fast}, there exists a schedule of length $O(C'+D')=O(\log n)$ for these time frames' \sus.
	Combining these schedules one time frame after another, we obtain a schedule of length 
	\begin{align*}
	\left(\frac{C+D}{\log n} + \log n\right)\cdot O(\log n) &= O(C+D+\log^2n).\qedhere
	\end{align*}
\end{proof}

\begin{figure}[h]
	\centering
	\begin{subfigure}[t]{0.32\textwidth}
		\centering
		\begin{tikzpicture}[scale=1]
		\drawUBScheduleA
		\end{tikzpicture}
		\caption{Wait for $T_i$'s random delay.}
	\end{subfigure}%
	\begin{subfigure}[t]{0.32\textwidth}
		\centering
		\begin{tikzpicture}[scale=1]
		\drawUBScheduleB
		\end{tikzpicture}
		\caption{Run LMR on $P_1$, $P_2$.}
	\end{subfigure}%
	\begin{subfigure}[t]{0.32\textwidth}
		\centering
		\begin{tikzpicture}[scale=1]
		\drawUBScheduleC
		\end{tikzpicture}
		\caption{Run LMR on $P_3, \ldots, P_8$.}
	\end{subfigure}%
	\caption{Our multicast schedule on $T_i$ using the decomposition from \Cref{fig:UBGraphDecomp}. Nodes with $m_i$ colored in black. Short unicast paths are dashed in black. ``LMR'' is the schedule given by \citet{leighton1994packet}.}
	\label{fig:UBSchedule}	
\end{figure}
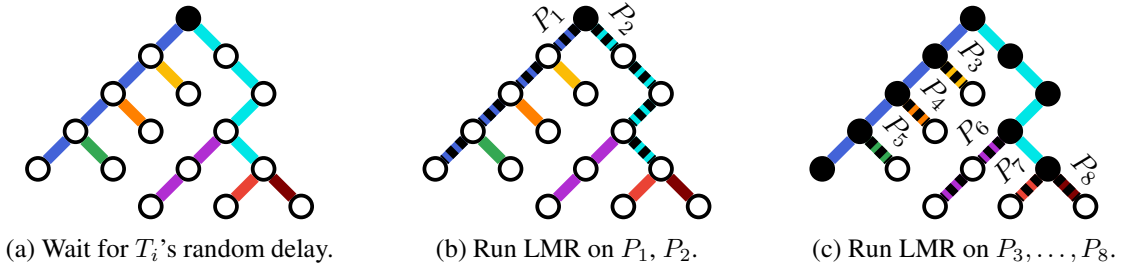

The above proof can be made algorithmic, deterministic, and even allows for efficient distributed algorithms. See \Cref{sec:algorithms} for details.

\section{Additive $\Omega(\log n)$ Necessary}\label{sec:AddLogn}
In this section we use our $\Omega(CD)$ lower bound (\Cref{lem:LB}) to show that any \sm bound of the form $O(C + D) + f(n)$ must have $f(n) = \Omega(\log n)$ (\Cref{lem:LBadd}). This result demonstrates the near optimality of the length $O(C + D + \log^2 n)$ schedules we gave in the previous section.
\lowerBoundAdd*
\begin{proof}[Proof of \Cref{lem:LBadd}]
	Assume for the sake of contradiction that every \sm instance admitted a schedule of length $\alpha(C + D) + f(n)$ for constant $\alpha$ and $f(n) = o(\log n)$. 
	
	Let $D = 4 \alpha$ and let $C = \frac{\log n}{2^{4\alpha + 1}}$ for $n$ to be fixed later. Consider the \sm instance given by $M_{\mcS}$ on graph $G_{\mcS}$ as defined in \Cref{sec:lowerBound} whose properties are given by \Cref{lem:LB}. Notice that $C2^{D+1} = \frac{\log n}{2^{4\alpha + 1}} 2^{4 \alpha + 1} = \log n$ and so indeed we may apply \Cref{lem:LB}.
	
	 Furthermore, by \Cref{lem:LB} we have that the optimal schedule of this \sm instance with congestion $C$, dilation $D$ and $n$ nodes has length at least
	\begin{align*}
	L := \frac{CD}{2} = \frac{\alpha}{2^{4\alpha}} \log n.
	\end{align*}
	
	But, by our assumption for contradiction we have that this instance admits a schedule of length at most
	\begin{align*}
	U :=	\alpha(C + D) + f(n) = 	\frac{\alpha}{2^{4\alpha + 1}}\log n + 4\alpha^2 + f(n).
	\end{align*}
	
	We then have a contradiction because $U <  L$. In particular for $n$ sufficiently large,
	\begin{align*}
	U - L & = -\frac{\alpha}{2^{4\alpha + 1}}\log n + 4\alpha^2 + o(\log n) < 0.\qedhere
	\end{align*}
\end{proof}

\section{Future Directions}
We conclude our paper with future directions for work in the scheduling of \sms. Of course, one can try and tighten the polylogarithmic additive terms in our results. More interestingly, one could extend the \sm setting in ways similar to how the \su scheduling work of \citet{leighton1994packet,leighton1999fast} has been extended. 

We give two notable examples. First, one could study what sort of approximation algorithms are possible if one is permitted to choose the trees over which multicast is performed as was done in the \su setting \cite{srinivasan2001constant,bertsimas1999asymptotically,koch2009real}. 
Roughly speaking, this corresponds to a depth-bounded version of the multicast congestion problem  \cite{vempala1999approximating,carr2002randomized,jansen2002approximation}.
We point out that choices of trees with optimal congestion + dilation (or nearly-optimal, up to constant multiplicative and additive polylogarithmic terms) combined with our algorithm to output length $O(C + D + \log^2 n)$-length schedules would imply near-optimal simultaneous multicasts for this setting.
Second, we note that our schedules have logarithmic-sized edge queues. That is, messages may have to wait up to $\Theta(\log n)$ rounds before being sent over an edge. This is not due to our use of the schedules of \citet{leighton1994packet}, whose queue sizes are constant, but rather due to $\Theta(\log n)$ messages arriving to a node by the end of simultaneous unicast frames used in our schedules. An interesting open question is whether there exist efficient \sm schedules which minimize both time and edges' queue sizes.


\appendix

\section{Deferred (Lower Bound) Content of \Cref{sec:lowerBound}}\label{app:lowerBound}

\lowerBoundIsMulticast*
\begin{proof}
	We prove this by induction on $D$. Let $T_\chi$ be the graph induced by label $\chi$ and let $r_\chi$ be $\chi$'s root. We will prove the slightly stronger claim that $T_\chi$ is a tree containing $r_\chi$ and for any two labels $\chi \neq \chi'$ we have $r_\chi \in T_{\chi'}$ only if $r_\chi = r_{\chi'}$. Call this latter property $\ostar$.
	
	As a base case suppose that $D = 1$. We then have that $\mcS = (S_1)$ is a single set of labels and so by definition of $M_\mcS$ we have that our graph will consist of a single edge $(r, v)$ where $r$ is the root for every label and $(r,v)$ is labeled by every label in $S_1$. Since each label induces the edge $(r,v)$ where $r$ is the root for this label, clearly every label induces a rooted tree containing its root. Moreover, $\ostar$ holds since every label has the same root.
	
	As an inductive hypothesis suppose that for any $D' < D$ and $\mcS'$ of size $C \cdot 2^{D'-1}$, we have that every label in $M_{\mcS'}$ induces a tree containing the label's root and all induced trees in $M_{\mcS'}$ satisfy $\ostar$. Thus, our inductive hypothesis tells us that every label in $M_{\mcS'}$ for $\mcS' \in I(\mcS)$ induces a tree containing the label's root where all labels satisfy $\ostar$.
	
	We will first verify that each $T_\chi$ is a tree containing $r_\chi$. Clearly, since the edge leaving $r_\chi$ is labeled $\chi$, we have that $r_\chi \in T\chi$. Let $\mcT_\chi'$ be all trees induced by label $\chi$ in $M_\mcS'$ for $\mcS' \in I(\mcS)$. $M_{\mcS}$ is created by taking the disjoint union of $G_{\mcS'}$ for $\mcS' \in I(\mcS)$, identifying several roots of $M_{\mcS'}$ trees and then adding new roots and edges. Notice that for any $\chi$ this identifying of nodes as the same nodes does not cause any cycles in a $T_\chi$ since by $\ostar$ we identify exactly one node from each tree in $\mcT_\chi'$ with another node. Next, to see that $\ostar$ still holds notice that if a node is designated a root in $M_\mcS$ then it is incident to a single edge and is a root for every label this edge was assigned.
\end{proof}

\correctCAndD*
\begin{proof}
	As each edge receives $C$ labels in our construction, each of which corresponds to a multicast tree, clearly the congestion is $C$.
	For the dilation, we prove by induction on $D$. As a base case notice that if $D$ is 1 then $|\mcS|$ is $1$ and so $G$ consists of a single edge used by all all trees, giving a dilation of $1$. Suppose that for $D' < D$ we have that the dilation of $M_{\mcS'}$ is $D'$ where $|\mcS'| = C \cdot 2^{D-1}$. The claim follows by simply noticing that each tree in $M_{\mcS}$ extends the root of every tree in $M_{\mcS'}$ by $1$ edge.
\end{proof}

\nUpperBound*
\begin{proof}
	
	Clearly, to upper bound the total number of vertices it suffices to upper bound the total number of edges  introduced. Thus, we will count the number of edges introduced at each level of our recursion. 
	
	Fix $C$. Define $m_D := |E(G_\mcS)|$. We claim by induction on $D$ that $ m_D \leq 2^{C(2^{D}) + D}$. As a base case notice that when $D = 1$ we have $m_1 = 1 \leq 2^{C(2^D) + D}$. For our inductive step consider $G_\mcS$. $G_\mcS$ is constructed by introducing $2^{D-1}$ edges and unioning together $G_{\mcS'}$ for $\mcS' \in I(\mcS)$ of which there are ${C \choose C/2}^{2^{D-1}}$, each of which have $m_{D-1}$ edges. Thus, we have 
	\begin{align*}
	m_D &= 2^{D-1} + {C \choose C/2}^{2^{D-1}} m_{D-1}\\
	& \leq 2^{D-1} + 2^{C2^{D-1}}2^{C2^{D-1} + D - 1} \bc{${C \choose C/2 } \leq 2^C$ and inductive hypothesis}\\
	& \leq 2^{D-1} + 2^{C2^{D} + D - 1}\\
	& \leq 2^{C2^D + D} \bc{$2^{D-1} \leq 2^{{C2^{D} + D - 1}}$}
	\end{align*}
	Finally, we conclude the (somewhat loose) bound of $2^{C2^{D+1}}$ on the number of vertices in our graph since $D \leq C2^D$.
\end{proof}

\section{Algorithmic Results}\label{sec:algorithms}

In this section we present centralized and distributed algorithms for the computation of \sm schedules of length $O(C+D+\log^2n)$, as guaranteed to exist by \Cref{thm:UBCentralized}. 

\subsection{Centralized Algorithm}\label{sec:centralized}

It is easy to see that the probabilistic method proof in \Cref{thm:UBCentralized} yields a randomized algorithm which succeeds with high probability. Moreover, by standard limited independence methods \cite{schmidt1995chernoff}, one can make this algorithm deterministic.

\begin{thm}
	There exists a deterministic, centralized algorithm which, given a \sm instance, outputs a schedule of length $O(C + D + \log^2 n)$ in time polynomial in $|\mcT|$ and $n$.
\end{thm}
\begin{proof}
	Let us begin by explaining why the proof of \Cref{thm:UBCentralized} immediately yields a polynomial-time randomized algorithm which succeeds with high probability. Recall that the schedules in \Cref{thm:UBCentralized} were produced by taking a $(\log n, \log n)$-short decomposition of each $T_i$, delaying each $T_i$ by $X_{T_i} \sim [C / \log n]$ and then concatenating together unicast schedules given by \citet{leighton1994packet}.  As noted in \Cref{sec:upper-bound},  a $(\log n, \log n)$-short decomposition can be computed deterministically in polynomial (in fact, linear) time. Clearly, drawing a random delay from $[C / \log n]$ for each $T_i$ is also doable in polynomial time by a randomized algorithm. Lastly, by \citet{leighton1999fast}, the schedules of \citet{leighton1994packet} can be computed deterministically in polynomial time. By \Cref{thm:UBCentralized} the resulting schedule is of the appropriate length.
		
	Let us now explain how this algorithm can be made deterministic. Let $n' = n + |\mcT|$. The only randomization used in the above algorithm is the random delays drawn from $[C / \log n]$. As with most proofs that show concentration by Chernoff bounds, it is easy to see that each $X_{T_i}$ need only be $\tfrac{1}{\poly n'}$-approximate, $O(\log n')$-wise independent for the above algorithm to succeed with high probability in $n'$. Recalling that one can generate polynomially-many binary $\tfrac{1}{\poly n'}$-approximate $O(\log n')$-wise independent random variables with only $O(\log n')$ random bits, our deterministic algorithm can simply brute force over all possible assignments to these $O(\log n')$ bits, and check if each resulting schedule is of the appropriate length. The result is a deterministic algorithm which is polynomial-time in $|\mcT|$ and $n$ and outputs a schedule of the appropriate length. For more background on limited independence, see \citet{schmidt1995chernoff}.
\end{proof}

\subsection{Distributed Algorithm}\label{sec:distributed}

In this section we give our distributed \sm algorithm in the \CONGEST model.

In the classic \CONGEST model of distributed communication \cite{peleg2000distributed}, a network is modeled as an undirected simple $n$-node graph $G=(V,E)$. Communication is conducted over discrete, synchronous rounds. During each round each node can send an $O(\log n)$-bit message along each of its incident edges. Every node has an arbitrary and unique ID of $O(\log n)$ bits, first only known to itself (this is the $KT_0$ model of \citet{awerbuch1990trade}).

In the \CONGEST model in a \sm instance, each node initially knows a unique ID associated with each tree $T_i$ to which it belongs, as well as which of its incident edges occur in which trees. We think of $m_i$ in this setting as being an $O(\log n)$-bit message, which is therefore transmittable along an edge in a single round. As in the centralized version of the problem, initially only $r_i$ knows $m_i$.

Our main result for this section is as follows.
\begin{restatable}{thm}{upperBoundDistributed}
	\label{thm:UBDistributed} 
	For any constant $\epsilon>0$, there exists a \CONGEST algorithm which given access to shared randomness solves \sm in time 
	\begin{align*}
	O\left((C+D)\cdot \left(1+\frac{\log \min \{C,D\}}{\log \log n}\right)+\log^{2+\epsilon}n\right).
	\end{align*}
	with high probability. If nodes also know their depth in each tree, then there exists another \CONGEST algorithm which solves \sm in $O(C+D + \log^{2+\epsilon})$ time.
\end{restatable}

As noted above, \sm has proven to be a crucial subroutine in many recent algorithms in \CONGEST for fundamental problems like MST, shortest path and approximate min cut. Therefore, improving simultaneous multicast in the \CONGEST model is an important step towards obtaining better algorithms for many of these fundamental problems.
Furthermore, in the above applications of \sm the parameters $C$ and $D$ are equal to the diameter of the graph up to polylogarithmic in $n$ terms provided the input graph has certain structure such as being planar \cite{ghaffari2016distributedII,haeupler2016low,ghaffari2018new,haeupler2016near}. If $C$ and $D$ are sufficiently large polylogarithmic terms, i.e., $\max\{C,D\}=\Omega(\log^{2+\epsilon}n)$, then, assuming nodes know their heights, our distributed algorithm gives an optimal $O(C+D)$ time distributed algorithm. Thus, we view our distributed algorithm as an important step towards obtaining better algorithms for many distributed problems, including MST, shortest path and approximate minimum cut.

Before, proceeding, let us discuss the preprocessing assumptions in \Cref{thm:UBDistributed}. Our distributed algorithms assume nodes have access to shared randomness or to their height in each of their incident multicast trees. Both of these assumptions can be dispensed with provided nodes are allowed to do some preprocessing: see \citet{ghaffari2015distributed} for how to share randomness and note that nodes can compute their heights by a single \sm computation where we could, for e.g.\ use the aforementioned $O(C + D \log n)$ length schedules. If this preprocessing is performed only once and many \sms are performed, its cost amortizes away.
Thus, provided nodes share randomness we have that after a preprocessing step equivalent to the current state of the art distributed \sm algorithm, subsequent \sms can be performed in time $O(C+D+\log^{2+\epsilon}n)$, which as discussed earlier, is essentially as close as one can get to a bound of $O(C + D)$.


\subsubsection{Intuition and Overview}
Before moving on, we will provide an intuition for and an overview of our results. As mentioned earlier, \citet{ostrovsky1997universal} provided a distributed algorithm for \su using $O(C + D + \log^{1+\epsilon} n)$ rounds. Since our centralized algorithm has shown that \sm can be reduced to \su by way of a $(\log n, \log n)$-short decomposition, the focus of our distributed algorithm is the efficient distributed computation of a $(\log n, \log n)$-short decomposition.

The challenge of computing such a decomposition in a distributed manner is that it seems as hard as solving simultaneous multicast. In particular, computing a heavy path decomposition requires that every node in a $T_i$ aggregate information from all of its children. It is not hard to see that performing such a ``convergecast'' at every node can be seen as performing a multicast on every $T_i$ in reverse. Even worse, the message size sent by nodes to their parents in such a convergecast to compute a heavy path decomposition must consist of  $\log_2 n$ bits to count the size of their sub-tree; i.e.\ sending just one such message fully uses the bandwidth of a \CONGEST link in one round. Thus, it seems that if we want to solve \sm by using a $(\log n, \log n)$-short decomposition, then we must circularly solve a simultaneous convergecast---i.e.\  \sm in reverse---in which large messages must be sent.

However, we show that, in fact, one can compute what is essentially a $(\log n, \log n)$-short decomposition more efficiently than one can solve \sm. In particular, we show how to efficiently compute what we call a $(\log ^{1 + \epsilon} n, \log n)$-short decomposition. We demonstrate that a $(\log ^{1 + \epsilon} n, \log n)$-short decomposition for every $T_i$ can be efficiently computed in a distributed fashion by using a ``rank-decomposition'' rather than a heavy path decomposition. Computing a rank-decomposition will require nodes to send exponentially fewer bits to their parents than computing a heavy path decomposition. By exploiting this exponential decrease in the total number of bits that must be passed, we are able to efficiently pack rank information into sent messages and compute a $(\log ^{1 + \epsilon} n, \log n)$-short decomposition. We are then able to translate our centralized algorithm to the distributed setting by making use of the distributed  \su algorithms of \citet{ostrovsky1997universal}.
\subsubsection{$(\log ^{1 + \epsilon} n, \log n)$-Short Decomposition Using a Rank Decomposition}
The decomposition which we compute will not be exactly identical to those of our centralized algorithms and so we generalize these decompositions as follows.

\begin{Def}\label{def:good-decomposition}
	For any integers $k$ and $\ell$, we say a path decomposition of a tree of depth $D$ is \emph{$(\ell,k)$-short} if each root-to-leaf path in the tree intersects at most $D/\ell+k$ paths of the decomposition.
\end{Def}
\noindent We note that some trees do not admit an $(O(\log n),k)$-short path decomposition with $k = \omega(\log_2n)$---for example, it is easy to see that a complete binary tree on $n$ nodes admits no such decomposition.


We now define our rank-based decompositions as follows; to our knowledge the notion of rank we use here first appeared in the union find data structure \cite{tarjan1984worst}. 

\begin{Def}[Rank-based path decomposition]
	A \emph{rank-based path decomposition} of a rooted tree $T$ is obtained in a bottom-up fashion, as follows. Each leaf $v$ has rank zero; i.e., $rank(v)=0$. 
	Each internal node $v$ with children set $child(v)$ has rank 
	\[
	rank(v) = \begin{cases}
	\max_{u\in child(v)}\{rank(u)\} & |\arg\max_{u\in child(v)}\{rank(u)\}|=1 \\
	\max_{u\in child(v)}\{rank(u)\} + 1 & else.
	\end{cases}
	\]
	Each non-leaf node selects one \emph{preferred} edge, which is an edge to a child of highest rank (breaking ties arbitrarily). We consider inclusion-wise maximal paths consisting of preferred edges, and for each highest node $v$ of such a path $p$, we add to the path $p$ the edge from $v$ to its ancestor (if any). The obtained paths form the rank-based path decomposition.
\end{Def}

As with heavy path decompositions, the above is clearly a path decomposition of the tree. 
Moreover, for this decomposition, too, each root-to-leaf path intersects at most $\log_2n$ paths of the decomposition, due to the following simple observation, which follows by induction on the nodes' heights.

\begin{obs}\label{rank-size}
	Each node $v$ of rank $i$ has at least $2^i$ descendants.
\end{obs}

Another consequence of \Cref{rank-size} is that no node has rank greater than $\log_2 n$. This in particular implies that nodes can send their rank information using only $\log \log n$ bits. This will prove useful when trying to compute these  rank-based decomposition, by relying on random offsetting and appropriate bit-packing. 

\begin{lem}\label{distributed-path-decomposition}
	For any $\epsilon\geq 0$, there exists a \CONGEST algorithm which, given a simultaneous multicast instance and shared randomness, computes a $(\log^{1+\epsilon}n,\log n)$-short path decomposition of the multicast trees with high probability in time
	\begin{align*}
	O\left((C+D)\cdot \left(1+\frac{\log \min \{C,D\}}{\log \log n}\right)\right).
	\end{align*}
\end{lem}

\begin{proof}
	We will first aggregate information needed to compute ranks, from which we compute a rank-based path decomposition. We then refine this decomposition to obtain the required $(\log^{1+\epsilon}n,\log n)$-short path decomposition.
	
	Our algorithm proceeds in time frames of $O\big(1+\frac{\log \min\{C,D\}}{\log \log n}\big)$ rounds each (to be specified below). During each time frame some nodes send messages in some of their trees to their parents in those tree, as follows. First, each tree $T$ has its leaves begin transmitting at some time frame $X_T$, where $X_T$ is a random integer in the range $[C]$. 
	Internal nodes transmit once they have received messages from all of their children. Whenever node $u$ transmits to its parent $v$ in a tree $T$ during some time frame, $u$ sends its rank in this tree.
	In addition, $u$ also it transmits some additional information which allows $v$ to map this rank information to the appropriate tree, as follows.
	If $C\leq D$, then $u$ transmits the  of tree $T$ among the (at most $C$) trees that contain edge $(u,v)$, using only $O(\log C)$ bits. Otherwise, $u$ also transmits the height of $u$ in $T$, denoted by $h_T(u)$. To see how the latter information allows $v$ to determine the ID of $T$, we note that a node $v$ receives a message from its child $u$ in tree $T$ in time frame $X_T+h_T(u)$, where $h_T(u)$ is the height of $u$ in $T$. 
	Therefore, as $v$ knows $X_T$ and receives $h_T(u)$, then if a single tree transmits along $(u,v)$ during that time frame, $v$ knows precisely which tree this is.
	To avoid ambiguity due to several trees $T_1,T_2,\dots$ having the same value of $X_T+h_T(u)$, the node $u$ sends its messages of $(rank_T(u), h_T(u))$ sorted by the IDs of $T$.
	Recalling that $u$ transmits to its parent $v$ in tree $T$ during time frame $X_T+h_T(u)$, and that $X_T$ is chosen uniformly in $[C]$, we find that each edge has messages sent up it by at most one tree in expectation at any given time frame. 
	Moreover, by standard concentration inequalities and union bound, there are at most $O\big(\frac{\log n}{\log \log n}\big)$ trees that use any given edge during during any time frame, w.h.p. 
	Therefore, as we can send $O(\log n)$ bits along any edge in one round, the $(\log \log n + \log \min\{C,D\})$-sized messages of all trees using this edge during that time frame can be sent (w.h.p.) in time frames of length $O\big(1+\frac{\log D}{\log \log n}\big)$ rounds. As we use $O\big(\frac{C}{\log^{1+\epsilon} n}+\frac{D}{\log^{1+\epsilon}n} + \log n\big)$ many time frames, this rank aggregation step takes 	$O\left((C+D)\cdot \left(1+\frac{\log \min \{C,D\}}{\log \log n}\right)\right)$ rounds.
	
	To obtain a rank-based path decomposition from this rank information, we spend a further $O(C)$ rounds
	after the ranks are computed, as follows. For each tree $T$, each node $v$ in $T$ waits $X_T$ rounds, after which it notifies each of its children $u$ in $T$ whether the edges $(u,v)$ is $v$'s preferred edge in $T$. As before, w.h.p., each edge $(u,v)$ has only $O\left(\frac{\log n}{\log \log n}\right)$ trees $T$ for which this (single-bit of) information needs to be sent, which can be performed in a single round (with the bits for each tree sorted by the ID of the tree).
	Therefore, these $C$ rounds suffice to compute a rank-based decomposition. The overall claimed running time follows.
	
	Finally, we compute the desired $(\log^{1+\epsilon}n,\log n)$-short decomposition by refining the rank-based decomposition in a top-down manner, as follows. 
	As before, we take time frames of length $O\left(1+\frac{\log \min\{C,D\}}{\log\log n}\right)$. This time, we have nodes transmit information downwards to their children, in reverse order relatively to when they received information from their children (this allows the children to determine to what tree each message corresponds).
	Specifically, nodes send information corresponding to the length of a short path (i.e., of length at most $\log^{1+\epsilon}n$), which will form part of the $(\log^{1+\epsilon}n,\log n)$-short decomposition. In addition, nodes send the relevant tree $T$'s ID, or their depth in $T$. As with our bottom-up subroutine, either message allows to determine $T$ (if we also sort the messages by the tree's IDs).
	The information corresponding to the length of a short path is zero if this edge is not preferred, or more generally if it is the first edge of a path of length $\log^{1+\epsilon}n$ in our more $(\log^{1+\epsilon}n,\log n)$-short path decomposition.  When a node $v$ in tree $T$ receives such length information $\ell$, it sends $\ell+1 \mod \log^{1+\epsilon} n$ to its child in $T$. 
	This information can be encoded using $O(\log \log n)$ bits. Consequently, as for the computation of ranks, the relevant $O\big(\frac{\log n}{\log\log n}\big)$ messages sent along any edge at any time frame (w.h.p.) can be sent during a time frame of $O\big(\frac{\log n}{\log\log n}\big)\cdot (\log \log n + \log\min\{C,D\}) = O\big(1+\frac{\log \min\{C,D\}}{\log\log n}\big)$ rounds.
	Therefore, as here too we use $O\big(\frac{C}{\log^{1+\epsilon} n}+\frac{D}{\log^{1+\epsilon}n} + \log n\big)$ many time frames, the desired $(\log^{1+\epsilon}n,\log n)$-short path decomposition is computed in the claimed $O\left((C+D)\cdot \left(1+\frac{\log \min \{C,D\}}{\log \log n}\right)\right)$ rounds.
\end{proof}

%

\subsubsection{Using Our $(\log ^{1 + \epsilon} n, \log n)$-Short Decompositions to Solve \SM}

Leveraging the distributed algorithm of \citet{ostrovsky1997universal} together with our $(\log^{1+\epsilon},\log n)$-short path decomposition of all trees, we immediately obtain a distributed schedule with similar running time to that of the algorithm in \Cref{distributed-path-decomposition}. 
Partitioning the trees further into subtrees of polylogarithmic depth, we even obtain a near $O(C+D)$ bound, provided nodes know their height in each tree. Concluding, we have the following theorem which gives the properties of our distributed algorithm.

\upperBoundDistributed*

\begin{proof}
	The first bound follows rather directly from \Cref{distributed-path-decomposition} and our remark following \Cref{sec:upper-bound}, whereby any $(\log^{1+\epsilon}n,\log n)$-short decomposition of the trees of a simultaneous multicast instance allows us to compute by a local-control algorithm (specifically, the algorithm of \citet{ostrovsky1997universal}), a schedule of length $C+D+\log^{2+\epsilon} n$. Since the time to compute the short path decomposition takes $O\left((C+D)\cdot \left(1+\frac{\log \min \{C,D\}}{\log \log n}\right)\right)$, the bound follows.
	
	Now, suppose all nodes know their depth in each multicast tree they belong to. 
	We perform the multicasts of time frames of length $O(\log^{2+\epsilon}n)$, as follows.
	For each tree, we divide the tree into subtrees, and in particular consider a partition of each tree into ranges of $L\triangleq \log^{2+\epsilon}n$ consecutive levels. Note that there are at most $D/L+1$ such levels per tree.
	Each tree $T$ will choose a random delay of $X_T$ time frames, chosen uniformly in $[\lceil C/L\rceil]$. Once 
	the time frame $X_T$ arrives, we will transmit down the first $L$ levels of the tree. In the next time frame we transmit this information down the following $L$ levels, and so on and so forth. All such transmissions during a time frame are an instance of \sm, but what are its congestion and dilation? The dilation here is trivially $D'=L$, by choice of levels. As for the congestion, since each edge belongs to $C$ trees and each tree delays the its transmissions by $X_T\sim_R [\lceil C/L\rceil]$ time frames, the congestion for each edge during any time frame is $O(L)$ in expectation and w.h.p.~(since $L=\log^{2+\epsilon}n)=\Omega(\log n)$). Therefore, by the first bound of this theorem, and since $\log D'=O(\log\log n)$, each such time frame's \sm instance can be scheduled in time 
	\begin{align*}
	O\left((C'+D')\cdot \left(1+\frac{\log \min \{C',D'\}}{\log \log n}\right) +\log^{2+\epsilon}n\right) = O(\log^{2+\epsilon}n) = O(L).
	\end{align*}
	Our algorithm runs for $\lceil C/L\rceil + D/L+1$ time frames (also accounting for the random delays). This \sm algorithm therefore takes time at most
	\begin{align*}
	\left(\frac{C}{L} + \frac{D}{L}+2\right)\cdot O(L) & = O(C+D+L) = O(C+D+\log^{2+\epsilon}n). \qedhere
	\end{align*}
\end{proof}

\bibliographystyle{plainnat}
\bibliography{abb,ultimate}

\begin{thebibliography}{50}
\providecommand{\natexlab}[1]{#1}
\providecommand{\url}[1]{\texttt{#1}}
\expandafter\ifx\csname urlstyle\endcsname\relax
  \providecommand{\doi}[1]{doi: #1}\else
  \providecommand{\doi}{doi: \begingroup \urlstyle{rm}\Url}\fi

\bibitem[Alon and Spencer(2016)]{alon2016probabilistic}
Noga Alon and Joel~H Spencer.
\newblock \emph{The probabilistic method}.
\newblock John Wiley \& Sons, 2016.

\bibitem[Alon et~al.(1991)Alon, Bar-Noy, Linial, and Peleg]{alon1991lower}
Noga Alon, Amotz Bar-Noy, Nathan Linial, and David Peleg.
\newblock A lower bound for radio broadcast.
\newblock \emph{Journal of Computer and System Sciences}, 43\penalty0
  (2):\penalty0 290--298, 1991.

\bibitem[Andrews et~al.(2010)Andrews, Chuzhoy, Guruswami, Khanna, Talwar, and
  Zhang]{andrews2010inapproximability}
Matthew Andrews, Julia Chuzhoy, Venkatesan Guruswami, Sanjeev Khanna, Kunal
  Talwar, and Lisa Zhang.
\newblock Inapproximability of edge-disjoint paths and low congestion routing
  on undirected graphs.
\newblock \emph{Combinatorica}, 30\penalty0 (5):\penalty0 485--520, 2010.

\bibitem[auf~der Heide and V{\"o}cking(1999)]{auf1999shortest}
Friedhelm~Meyer auf~der Heide and Berthold V{\"o}cking.
\newblock Shortest-path routing in arbitrary networks.
\newblock \emph{Journal of Algorithms}, 31\penalty0 (1):\penalty0 105--131,
  1999.

\bibitem[Awerbuch et~al.(1990)Awerbuch, Goldreich, Vainish, and
  Peleg]{awerbuch1990trade}
Baruch Awerbuch, Oded Goldreich, Ronen Vainish, and David Peleg.
\newblock A trade-off between information and communication in broadcast
  protocols.
\newblock \emph{Journal of the ACM (JACM)}, 37\penalty0 (2):\penalty0 238--256,
  1990.

\bibitem[Bar-Noy et~al.(2000)Bar-Noy, Guha, Naor, and Schieber]{bar2000message}
Amotz Bar-Noy, Sudipto Guha, Joseph Naor, and Baruch Schieber.
\newblock Message multicasting in heterogeneous networks.
\newblock \emph{SIAM Journal on Computing (SICOMP)}, 30\penalty0 (2):\penalty0
  347--358, 2000.

\bibitem[Bertsimas and Gamarnik(1999)]{bertsimas1999asymptotically}
Dimitris Bertsimas and David Gamarnik.
\newblock Asymptotically optimal algorithms for job shop scheduling and packet
  routing.
\newblock \emph{Journal of Algorithms}, 33\penalty0 (2):\penalty0 296--318,
  1999.

\bibitem[Busch et~al.(2004)Busch, Magdon-Ismail, Mavronicolas, and
  Spirakis]{busch2004direct}
Costas Busch, Malik Magdon-Ismail, Marios Mavronicolas, and Paul Spirakis.
\newblock Direct routing: Algorithms and complexity.
\newblock In \emph{Proceedings of the 12th Annual European Symposium on
  Algorithms (ESA)}, pages 134--145, 2004.

\bibitem[Carr and Vempala(2002)]{carr2002randomized}
Robert Carr and Santosh Vempala.
\newblock Randomized metarounding.
\newblock \emph{Random Structures \& Algorithms}, 20\penalty0 (3):\penalty0
  343--352, 2002.

\bibitem[Chu et~al.(2001)Chu, Rao, Seshan, and Zhang]{chu2001enabling}
Yang Chu, Sanjay Rao, Srinivasan Seshan, and Hui Zhang.
\newblock Enabling conferencing applications on the internet using an overlay
  muilticast architecture.
\newblock \emph{ACM SIGCOMM computer communication review}, 31\penalty0
  (4):\penalty0 55--67, 2001.

\bibitem[Chuzhoy(2012)]{chuzhoy2012routing}
Julia Chuzhoy.
\newblock Routing in undirected graphs with constant congestion.
\newblock In \emph{Proceedings of the 44th Annual ACM Symposium on Theory of
  Computing (STOC)}, pages 855--874, 2012.

\bibitem[Chuzhoy et~al.(2007)Chuzhoy, Guruswami, Khanna, and
  Talwar]{chuzhoy2007hardness}
Julia Chuzhoy, Venkatesan Guruswami, Sanjeev Khanna, and Kunal Talwar.
\newblock Hardness of routing with congestion in directed graphs.
\newblock In \emph{Proceedings of the 39th Annual ACM Symposium on Theory of
  Computing (STOC)}, pages 165--178, 2007.

\bibitem[Chuzhoy et~al.(2017)Chuzhoy, Kim, and Nimavat]{chuzhoy2017new}
Julia Chuzhoy, David~HK Kim, and Rachit Nimavat.
\newblock New hardness results for routing on disjoint paths.
\newblock In \emph{Proceedings of the 49th Annual ACM Symposium on Theory of
  Computing (STOC)}, pages 86--99, 2017.

\bibitem[Deering(1988)]{deering1988host}
Stephen~E Deering.
\newblock Host extensions for {IP} multicasting.
\newblock \emph{Stanford University Memo}, 1988.

\bibitem[Dory and Ghaffari(2019)]{dory2019improved}
Michal Dory and Mohsen Ghaffari.
\newblock Improved distributed approximations for minimum-weight
  two-edge-connected spanning subgraph.
\newblock In \emph{Proceedings of the 38th ACM Symposium on Principles of
  Distributed Computing (PODC)}, pages 521--530, 2019.

\bibitem[Elkin and Kortsarz(2003)]{elkin2003sublogarithmic}
Michael Elkin and Guy Kortsarz.
\newblock Sublogarithmic approximation for telephone multicast: path out of
  jungle.
\newblock In \emph{Proceedings of the 14th Annual ACM-SIAM Symposium on
  Discrete Algorithms (SODA)}, pages 76--85, 2003.

\bibitem[Gaber and Mansour(1995)]{gaber1995broadcast}
Iris Gaber and Yishay Mansour.
\newblock Broadcast in radio networks.
\newblock In \emph{Proceedings of the 6th Annual ACM-SIAM Symposium on Discrete
  Algorithms (SODA)}, volume~95, pages 577--582. Citeseer, 1995.

\bibitem[Ghaffari(2015{\natexlab{a}})]{ghaffari2015distributed}
Mohsen Ghaffari.
\newblock Distributed broadcast revisited: Towards universal optimality.
\newblock In \emph{Proceedings of the 42nd International Colloquium on
  Automata, Languages and Programming (ICALP)}, pages 638--649. Springer,
  2015{\natexlab{a}}.

\bibitem[Ghaffari(2015{\natexlab{b}})]{ghaffari2015nearb}
Mohsen Ghaffari.
\newblock Near-optimal scheduling of distributed algorithms.
\newblock In \emph{Proceedings of the 34th ACM Symposium on Principles of
  Distributed Computing (PODC)}, pages 3--12, 2015{\natexlab{b}}.

\bibitem[Ghaffari and Haeupler(2016{\natexlab{a}})]{ghaffari2016distributedI}
Mohsen Ghaffari and Bernhard Haeupler.
\newblock Distributed algorithms for planar networks i: Planar embedding.
\newblock In \emph{Proceedings of the 35th ACM Symposium on Principles of
  Distributed Computing (PODC)}, pages 29--38, 2016{\natexlab{a}}.

\bibitem[Ghaffari and Haeupler(2016{\natexlab{b}})]{ghaffari2016distributedII}
Mohsen Ghaffari and Bernhard Haeupler.
\newblock Distributed algorithms for planar networks {II}: Low-congestion
  shortcuts, mst, and min-cut.
\newblock In \emph{Proceedings of the 27th Annual ACM-SIAM Symposium on
  Discrete Algorithms (SODA)}, pages 202--219, 2016{\natexlab{b}}.

\bibitem[Ghaffari and Li(2018)]{ghaffari2018new}
Mohsen Ghaffari and Jason Li.
\newblock New distributed algorithms in almost mixing time via transformations
  from parallel algorithms.
\newblock In \emph{Proceedings of the 32nd International Symposium on
  Distributed Computing (DISC)}, pages 31:1--31:16, 2018.

\bibitem[Haeupler et~al.(2016{\natexlab{a}})Haeupler, Izumi, and
  Zuzic]{haeupler2016low}
Bernhard Haeupler, Taisuke Izumi, and Goran Zuzic.
\newblock Low-congestion shortcuts without embedding.
\newblock In \emph{Proceedings of the 35th ACM Symposium on Principles of
  Distributed Computing (PODC)}, pages 451--460, 2016{\natexlab{a}}.

\bibitem[Haeupler et~al.(2016{\natexlab{b}})Haeupler, Izumi, and
  Zuzic]{haeupler2016near}
Bernhard Haeupler, Taisuke Izumi, and Goran Zuzic.
\newblock Near-optimal low-congestion shortcuts on bounded parameter graphs.
\newblock In \emph{Proceedings of the 30th International Symposium on
  Distributed Computing (DISC)}, pages 158--172, 2016{\natexlab{b}}.

\bibitem[Haeupler et~al.(2018{\natexlab{a}})Haeupler, Hershkowitz, and
  Wajc]{haeupler2018round}
Bernhard Haeupler, D~Ellis Hershkowitz, and David Wajc.
\newblock Round-and message-optimal distributed graph algorithms.
\newblock In \emph{Proceedings of the 37th ACM Symposium on Principles of
  Distributed Computing (PODC)}, pages 119--128, 2018{\natexlab{a}}.

\bibitem[Haeupler et~al.(2018{\natexlab{b}})Haeupler, Li, and
  Zuzic]{haeupler2018minor}
Bernhard Haeupler, Jason Li, and Goran Zuzic.
\newblock Minor excluded network families admit fast distributed algorithms.
\newblock In \emph{Proceedings of the 37th ACM Symposium on Principles of
  Distributed Computing (PODC)}, pages 465--474, 2018{\natexlab{b}}.

\bibitem[Holzer and Wattenhofer(2012)]{holzer2012optimal}
Stephan Holzer and Roger Wattenhofer.
\newblock Optimal distributed all pairs shortest paths and applications.
\newblock In \emph{Proceedings of the 31st ACM Symposium on Principles of
  Distributed Computing (PODC)}, pages 355--364, 2012.

\bibitem[Iglesias et~al.(2015)Iglesias, Rajaraman, Ravi, and
  Sundaram]{iglesias2015rumors}
Jennifer Iglesias, Rajmohan Rajaraman, R~Ravi, and Ravi Sundaram.
\newblock Rumors across radio, wireless, telephone.
\newblock In \emph{35th IARCS Annual Conference on Foundations of Software
  Technology and Theoretical Computer Science (FSTTCS)}, 2015.

\bibitem[Iglesias et~al.(2018)Iglesias, Rajaraman, Ravi, and
  Sundaram]{iglesias2018plane}
Jennifer Iglesias, Rajmohan Rajaraman, R~Ravi, and Ravi Sundaram.
\newblock Plane gossip: Approximating rumor spread in planar graphs.
\newblock In \emph{Proceedings of the 13th Latin American Theoretical
  Informatics Symposium (LATIN)}, pages 611--624, 2018.

\bibitem[Jannotti et~al.(2000)Jannotti, Gifford, Johnson, Kaashoek,
  et~al.]{jannotti2000overcast}
John Jannotti, David~K Gifford, Kirk~L Johnson, M~Frans Kaashoek, et~al.
\newblock Overcast: reliable multicasting with on overlay network.
\newblock In \emph{Proceedings of the 4th USENIX Symposium on Operating Systems
  Design and Implementation (OSDI)}, page~14, 2000.

\bibitem[Jansen and Zhang(2002)]{jansen2002approximation}
K~Jansen and H~Zhang.
\newblock An approximation algorithm for the multicast congestion problem via
  minimum steiner trees.
\newblock In \emph{Proceedings of the 3rd International Workshop on
  Approximation Algorithms for Combinatorial Optimization Problems (APPROX)},
  pages 77--90, 2002.

\bibitem[Koch et~al.(2009)Koch, Peis, Skutella, and Wiese]{koch2009real}
Ronald Koch, Britta Peis, Martin Skutella, and Andreas Wiese.
\newblock Real-time message routing and scheduling.
\newblock In \emph{Approximation, Randomization, and Combinatorial
  Optimization. Algorithms and Techniques}, pages 217--230. Springer, 2009.

\bibitem[Leighton et~al.(1994)Leighton, Maggs, and Rao]{leighton1994packet}
Thomson Leighton, Bruce~M Maggs, and Satish~B Rao.
\newblock Packet routing and job-shop scheduling in ${O}$(congestion+ dilation)
  steps.
\newblock \emph{Combinatorica}, 14\penalty0 (2):\penalty0 167--186, 1994.

\bibitem[Leighton et~al.(1999)Leighton, Maggs, and Richa]{leighton1999fast}
Tom Leighton, Bruce Maggs, and Andrea~W Richa.
\newblock Fast algorithms for finding o (congestion+ dilation) packet routing
  schedules.
\newblock \emph{Combinatorica}, 19\penalty0 (3):\penalty0 375--401, 1999.

\bibitem[Lenzen and Peleg(2013)]{lenzen2013efficient}
Christoph Lenzen and David Peleg.
\newblock Efficient distributed source detection with limited bandwidth.
\newblock In \emph{Proceedings of the 32nd ACM Symposium on Principles of
  Distributed Computing (PODC)}, pages 375--382, 2013.

\bibitem[Meyer and V{\"o}cking(1995)]{meyer1995packet}
Friedhelm Meyer and Berthold V{\"o}cking.
\newblock A packet routing protocol for arbitrary networks.
\newblock In \emph{Proceedings of the 12th International Symposium on
  Theoretical Aspects of Computer Science (STACS)}, pages 291--302. Springer,
  1995.

\bibitem[Ostrovsky and Rabani(1997)]{ostrovsky1997universal}
Rafail Ostrovsky and Yuval Rabani.
\newblock Universal ${O} (congestion+ dilation+ \log^{1+\epsilon}n)$ local
  control packet switching algorithms.
\newblock In \emph{Proceedings of the 29th Annual ACM Symposium on Theory of
  Computing (STOC)}, volume~29, pages 644--653, 1997.

\bibitem[Panagiotou et~al.(2015)Panagiotou, Perez-Gimenez, Sauerwald, and
  Sun]{panagiotou2015randomized}
Konstantinos Panagiotou, Xavier Perez-Gimenez, Thomas Sauerwald, and He~Sun.
\newblock Randomized rumour spreading: The effect of the network topology.
\newblock \emph{Combinatorics, Probability and Computing}, 24\penalty0
  (2):\penalty0 457--479, 2015.

\bibitem[Peis and Wiese(2011)]{peis2011universal}
Britta Peis and Andreas Wiese.
\newblock Universal packet routing with arbitrary bandwidths and transit times.
\newblock In \emph{Proceedings of the 13th Conference on Integer Programming
  and Combinatorial Optimization (IPCO)}, pages 362--375, 2011.

\bibitem[Peis et~al.(2009)Peis, Skutella, and Wiese]{peis2009packet}
Britta Peis, Martin Skutella, and Andreas Wiese.
\newblock Packet routing: Complexity and algorithms.
\newblock In \emph{Proceedings of the 7th Workshop on Approximation and Online
  Algorithms (WAOA)}, pages 217--228, 2009.

\bibitem[Peleg(2000)]{peleg2000distributed}
David Peleg.
\newblock Distributed computing.
\newblock \emph{SIAM Monographs on discrete mathematics and applications},
  5:\penalty0 1--1, 2000.

\bibitem[Rabani and Tardos(1996)]{rabani1996distributed}
Yuval Rabani and {\'E}va Tardos.
\newblock Distributed packet switching in arbitrary networks.
\newblock In \emph{Proceedings of the 28th Annual ACM Symposium on Theory of
  Computing (STOC)}, volume~96, pages 366--375, 1996.

\bibitem[Rothvo{\ss}(2013)]{rothvoss2013simpler}
Thomas Rothvo{\ss}.
\newblock A simpler proof for ${O} ({C}ongestion + {D}ilation)$ packet routing.
\newblock In \emph{Proceedings of the 16th Conference on Integer Programming
  and Combinatorial Optimization (IPCO)}, pages 336--348, 2013.

\bibitem[Scheideler(2006)]{scheideler2006universal}
Christian Scheideler.
\newblock \emph{Universal routing strategies for interconnection networks},
  volume 1390.
\newblock Springer, 2006.

\bibitem[Schmidt et~al.(1995)Schmidt, Siegel, and
  Srinivasan]{schmidt1995chernoff}
Jeanette~P Schmidt, Alan Siegel, and Aravind Srinivasan.
\newblock Chernoff--hoeffding bounds for applications with limited
  independence.
\newblock \emph{SIAM Journal on Discrete Mathematics}, 8\penalty0 (2):\penalty0
  223--250, 1995.

\bibitem[Sleator and Tarjan(1983)]{sleator1983data}
Daniel~D Sleator and Robert~Endre Tarjan.
\newblock A data structure for dynamic trees.
\newblock \emph{Journal of Computer and System Sciences}, 26\penalty0
  (3):\penalty0 362--391, 1983.

\bibitem[Srinivasan and Teo(2001)]{srinivasan2001constant}
Aravind Srinivasan and Chung-Piaw Teo.
\newblock A constant-factor approximation algorithm for packet routing and
  balancing local vs. global criteria.
\newblock \emph{SIAM Journal on Computing (SICOMP)}, 30\penalty0 (6):\penalty0
  2051--2068, 2001.

\bibitem[Tarjan and Van~Leeuwen(1984)]{tarjan1984worst}
Robert~E Tarjan and Jan Van~Leeuwen.
\newblock Worst-case analysis of set union algorithms.
\newblock \emph{Journal of the ACM (JACM)}, 31\penalty0 (2):\penalty0 245--281,
  1984.

\bibitem[Topkis(1985)]{topkis1985concurrent}
Donald~M. Topkis.
\newblock Concurrent broadcast for information dissemination.
\newblock \emph{IEEE Transactions on Software Engineering}, SE-11\penalty0
  (10):\penalty0 1107--1112, 1985.

\bibitem[Vempala and V{\"o}cking(1999)]{vempala1999approximating}
Santosh Vempala and Berthold V{\"o}cking.
\newblock Approximating multicast congestion.
\newblock In \emph{Proceedings of the 10th Annual International Symposium on
  Algorithms and Computation (ISAAC)}, pages 367--372, 1999.

\end{thebibliography}
	
\end{document}